%% file: arxiv.tex
\documentclass{article}

\usepackage{arxiv}

\usepackage[utf8]{inputenc} % allow utf-8 input
\usepackage[T1]{fontenc}    % use 8-bit T1 fonts
\usepackage{hyperref}       % hyperlinks
\usepackage{url}            % simple URL typesetting
\usepackage{booktabs}       % professional-quality tables
\usepackage{amsfonts}       % blackboard math symbols
\usepackage{nicefrac}       % compact symbols for 1/2, etc.
\usepackage{microtype}      % microtypography
\usepackage{graphicx}
\usepackage[numbers, compress]{natbib}
\usepackage{doi}

\usepackage{orcidlink}
\usepackage{algorithm}
\usepackage{algpseudocode}
\usepackage{amsmath, amssymb, amsthm}
\usepackage{subcaption}
\usepackage[normalem]{ulem} % for \sout
\usepackage{multirow}
\usepackage{xcolor}         % colors
\usepackage{makecell}
\usepackage{empheq}
\usepackage[most]{tcolorbox}
\usepackage{wrapfig}

\algnewcommand\Input{\item[\textbf{Input:}]}
\algnewcommand\Output{\item[\textbf{Output:}]}

\theoremstyle{plain}
\newtheorem{theorem}{Theorem}[section]
\newtheorem{proposition}[theorem]{Proposition}

\newtheorem{corollary}[theorem]{Corollary}
\theoremstyle{definition}

\newtheorem{assumption}[theorem]{Assumption}
\theoremstyle{remark}
\newtheorem{remark}[theorem]{Remark}

\newtcbox{\eqbox}{
  colback=blue!15,
  colframe=black,
  arc=3mm,
  boxrule=0.8pt,
  left=6pt,
  right=6pt,
  top=6pt,
  bottom=6pt
}

\title{Nonparametric Deconvolution and Denoising using Simulation Based Inference}
\renewcommand{\shorttitle}{Nonparametric Deconvolution and Denoising} 

\author{
  Ritwik~Vashistha \orcidlink{tbd}\\
  Department of Statistics and Data Sciences, University of Texas at Austin, Austin, Texas 78712, USA\\
  The NSF-Simons AI Institute for Cosmic Origins, USA \\
  \texttt{ritwik.v@utexas.edu} \\
  \And
  Abhra~Sarkar \orcidlink{tbd}\\
  Department of Statistics and Data Sciences, University of Texas at Austin, Austin, Texas 78712, USA\\
  \texttt{abhra.sarkar@utexas.edu}\\
  \And
  Arya~Farahi \orcidlink{0000-0003-0777-4618}\\
  Department of Statistics and Data Sciences, University of Texas at Austin, Austin, Texas 78712, USA\\
  The NSF-Simons AI Institute for Cosmic Origins, USA \\
  \texttt{arya.farahi@utexas.edu} \\
}

\date{} 

\begin{document}

\maketitle

\begin{abstract}
Latent signals are often obscured by measurement noise, yet encode the underlying laws and dynamics of complex systems; learning both the signals and their distributions remains a central challenge in scientific inference. The noise is often non-negligible, and the likelihoods for expressive generative models are often intractable. We utilize a convolutional maximum mean discrepancy (convMMD) loss and propose a likelihood-free framework for nonparametric density deconvolution and empirical Bayes denoising under additive measurement error. Our method learns a latent generative model by matching the observed data distribution to the noise-convolved model distribution. This yields a differentiable, simulation-based objective for multivariate homoscedastic or heteroscedastic noise, compatible with expressive sieve classes such as Gaussian mixtures and normalizing flows. The learned density then serves as an empirical prior for posterior denoising of individual latent values. Theoretically, we extend convMMD from parametric to nonparametric estimation, proving finite-sample bounds for empirical sieve minimizers and $L_2$ convergence rates under Sobolev smoothness. These rates recover the classical inverse-problem dependence: polynomial for ordinary-smooth and logarithmic for super-smooth noises. Our method provides a practical, theoretically grounded approach to deconvolution and denoising under generative latent distribution models.
\end{abstract}

\input{sections/introduction}

\input{sections/theory}

\input{sections/results}

\input{sections/conclusion}

\section*{Acknowledgments}
We acknowledge support from the National Science Foundation under Cooperative Agreement 2421782 and grant DMS-2515902, and the Simons Foundation award MPS-AI-00010515.

\newpage

\bibliographystyle{unsrt}
\bibliography{references}

\newpage
%%%%%%%%%%%%%%%%%%%%%%%%%%%%%%%%%%%%%%%%%%%%%%%%%%%%%%%%%%%%

\appendix

\input{sections/assumptions}

\input{sections/justifications}

%%%%%%%%%%%%%%%%%%%%%%%%%%%%%%%%%%%%%%%%%%%%%%%%%%%%%%%%%%%%

\subsection{Proofs}

\begin{theorem}
Let $k: \mathbb{R}^d \times \mathbb{R}^d \rightarrow \mathbb{R}$ be a continuous, translation-invariant kernel bounded by $\sup_{\mathbf{x}} k(\mathbf{x}, \mathbf{x}) \leq K$. Let $\widehat{(p * m)}_N$ denote the empirical distribution of $N$ i.i.d. observations drawn from the true noisy distribution $p * m$. Let the estimator $\widehat{\boldsymbol{\theta}}_N$ be the global minimizer of the exact empirical squared objective over the sieve class $\Theta_J$:  
$$
\widehat{\boldsymbol{\theta}}_N=\underset{\boldsymbol{\theta} \in \Theta_J}{\arg \min }\ \widehat{\operatorname{convMMD}}_k^2\left(p, q_{\boldsymbol{\theta}}, m\right) \equiv \underset{\boldsymbol{\theta} \in \Theta_J}{\arg \min }\ \operatorname{MMD}_k^2\left(\widehat{(p * m)}_N, q_{\boldsymbol{\theta}} * m\right).
$$
Then, with probability at least $1-\delta$, we have:
$$
\operatorname{convMMD}_k\left(p, q_{\widehat{\boldsymbol{\theta}}_N}, m\right) \leq \inf_{\boldsymbol{\theta} \in \Theta_J} \operatorname{convMMD}_k\left(p, q_{\boldsymbol{\theta}}, m\right)+2 \sqrt{\frac{K}{N}}+2 \sqrt{\frac{2 K \log (1 / \delta)}{N}}.
$$
\end{theorem}
\begin{proof}
Let $\mathcal{H}_k$ be the Reproducing Kernel Hilbert Space (RKHS) associated with kernel $k$. Let $\mu_P \in \mathcal{H}_k$ denote the mean embedding of a distribution $P$, defined as $\mu_P(\cdot) = \mathbb{E}_{\mathbf{x} \sim P}[k(\mathbf{x}, \cdot)]$.
The MMD between two distributions $P$ and $Q$ is then simply the distance between their mean embeddings in the RKHS:
$$ \text{MMD}(P, Q) = \| \mu_P - \mu_Q \|_{\mathcal{H}_k}. $$

By definition of the convMMD, $\text{convMMD}_k(p, q_{\boldsymbol{\theta}}, m) = \text{MMD}_k(p*m, q_{\boldsymbol{\theta}}*m) = \| \mu_{p*m} - \mu_{q_{\boldsymbol{\theta}}*m} \|_{\mathcal{H}_k}$.

Because the scalar mapping $v \mapsto v^2$ is strictly increasing on $[0, \infty)$, the estimator $\widehat{\boldsymbol{\theta}}_N$ that minimizes the squared empirical MMD identically minimizes the unsquared empirical MMD. Thus, for any arbitrary $\boldsymbol{\theta} \in \Theta_J$:

$$\| \mu_{\widehat{(p*m)}_N} - \mu_{q_{\widehat{\boldsymbol{\theta}}_N}*m} \|_{\mathcal{H}_k} \leq \| \mu_{\widehat{(p*m)}_N} - \mu_{q_{\boldsymbol{\theta}}*m} \|_{\mathcal{H}_k}.$$
By the triangle inequality in $\mathcal{H}_k$, we bound the true error of our estimator by introducing the empirical data embedding:

$$\| \mu_{p*m} - \mu_{q_{\widehat{\boldsymbol{\theta}}_N}*m} \|_{\mathcal{H}_k} \leq \| \mu_{p*m} - \mu_{\widehat{(p*m)}_N} \|_{\mathcal{H}_k} + \| \mu_{\widehat{(p*m)}_N} - \mu_{q_{\widehat{\boldsymbol{\theta}}_N}*m} \|_{\mathcal{H}_k}.$$
Applying the optimality condition of $\widehat{\boldsymbol{\theta}}_N$ to the second term:

$$\| \mu_{p*m} - \mu_{q_{\widehat{\boldsymbol{\theta}}_N}*m} \|_{\mathcal{H}_k} \leq \| \mu_{p*m} - \mu_{\widehat{(p*m)}_N} \|_{\mathcal{H}_k} + \| \mu_{\widehat{(p*m)}_N} - \mu_{q_{\boldsymbol{\theta}}*m} \|_{\mathcal{H}_k}.$$
We apply the triangle inequality once more to the right-most term to re-introduce the true target embedding $\mu_{p*m}$:

$$\| \mu_{\widehat{(p*m)}_N} - \mu_{q_{\boldsymbol{\theta}}*m} \|_{\mathcal{H}_k} \leq \| \mu_{\widehat{(p*m)}_N} - \mu_{p*m} \|_{\mathcal{H}_k} + \| \mu_{p*m} - \mu_{q_{\boldsymbol{\theta}}*m} \|_{\mathcal{H}_k}.$$

Combining these yields the fundamental decomposition. Since the bound holds for any $\boldsymbol{\theta} \in \Theta_J$, we take the infimum over the sieve:

\begin{equation} \label{eq: decompostion}
\| \mu_{p*m} - \mu_{q_{\widehat{\boldsymbol{\theta}}_N}*m} \|_{\mathcal{H}_k} \leq \underbrace{\inf_{\boldsymbol{\theta} \in \Theta_J} \| \mu_{p*m} - \mu_{q_{\boldsymbol{\theta}}*m} \|_{\mathcal{H}_k}}_{\text{Approximation Error}} + \underbrace{2 \| \mu_{p*m} - \mu_{\widehat{(p*m)}_N} \|_{\mathcal{H}_k}}_{\text{Estimation Error}}.
\end{equation}

Since GMMs are dense in $L_1\left(\mathbb{R}^d\right)$ \citep{goodfellow2016}, and convolution with $m$ is a continuous linear operator, the set of convolved models $\left\{q_{\boldsymbol{\theta}} * m\right\}$ is dense in the space of observed densities.

We bound the statistical estimation error $\Delta(\widetilde{\mathbf{x}}_1, \dots, \widetilde{\mathbf{x}}_N) = \| \mu_{p*m} - \mu_{\widehat{(p*m)}_N} \|_{\mathcal{H}_k}$ using McDiarmid's Inequality \citep{mcdiarmid1989method}. 
If we replace a single observation $\widetilde{\mathbf{x}}_i$ with an independent draw $\widetilde{\mathbf{x}}_i'$, the value of $\Delta$ changes by at most:
$$| \Delta(\dots, \widetilde{\mathbf{x}}_i, \dots) - \Delta(\dots, \widetilde{\mathbf{x}}_i', \dots) | \leq \frac{1}{N} \| k(\widetilde{\mathbf{x}}_i, \cdot) - k(\widetilde{\mathbf{x}}_i', \cdot) \|_{\mathcal{H}_k}.$$
Using the property of the RKHS norm:
$$\| k(\widetilde{\mathbf{x}}_i, \cdot) - k(\widetilde{\mathbf{x}}_i', \cdot) \|_{\mathcal{H}_k} = \sqrt{k(\widetilde{\mathbf{x}}_i, \widetilde{\mathbf{x}}_i) + k(\widetilde{\mathbf{x}}_i', \widetilde{\mathbf{x}}_i') - 2k(\widetilde{\mathbf{x}}_i, \widetilde{\mathbf{x}}_i')}.$$
By the Cauchy-Schwarz inequality in $\mathcal{H}_k$, we have $|k(\widetilde{\mathbf{x}}_i, \widetilde{\mathbf{x}}_i')| \leq \sqrt{k(\widetilde{\mathbf{x}}_i, \widetilde{\mathbf{x}}_i) k(\widetilde{\mathbf{x}}_i', \widetilde{\mathbf{x}}_i')} \leq K$, meaning $-k(\widetilde{\mathbf{x}}_i, \widetilde{\mathbf{x}}_i') \le K$. Thus:
$$\| k(\widetilde{\mathbf{x}}_i, \cdot) - k(\widetilde{\mathbf{x}}_i', \cdot) \|_{\mathcal{H}_k} \leq \sqrt{2K - 2(-K)} = 2\sqrt{K}.$$
Thus, the maximum change is bounded by $c_i = \frac{2\sqrt{K}}{N}$.
Applying McDiarmid's Inequality, for any $\epsilon > 0$:

$$\mathbb{P}(\Delta - \mathbb{E}[\Delta] \geq \epsilon) \leq \exp \left( \frac{-2\epsilon^2}{\sum_{i=1}^N c_i^2} \right) = \exp \left( \frac{-2\epsilon^2}{N (\frac{4K}{N^2})} \right) = \exp \left( \frac{-N\epsilon^2}{2K} \right).$$
Setting $\delta = \exp(-N\epsilon^2 / 2K)$ and solving for $\epsilon$ gives, with probability at least $1-\delta$:

\begin{equation} \label{eq: mc-darmid}
\Delta \leq \mathbb{E}[\Delta] + \sqrt{\frac{2K \log(1/\delta)}{N}}.
\end{equation}

By Jensen's inequality, $\mathbb{E}[\Delta] \leq \sqrt{\mathbb{E}[\Delta^2]}$. Expanding the squared RKHS norm gives the variance of the sample mean embedding:

$$\mathbb{E}[\Delta^2] = \mathbb{E} \left[ \left\| \frac{1}{N} \sum_{i=1}^N k(\widetilde{\mathbf{x}}_i, \cdot) - \mu_{p*m} \right\|_{\mathcal{H}_k}^2 \right] = \frac{1}{N} \mathbb{E}_{\widetilde{\mathbf{X}} \sim p*m} \left[ \| k(\widetilde{\mathbf{X}}, \cdot) - \mu_{p*m} \|_{\mathcal{H}_k}^2 \right].$$
Using the identity $\mathbb{E}[\|\mathbf{X} - \mathbb{E}[\mathbf{X}]\|^2] \leq \mathbb{E}[\|\mathbf{X}\|^2]$, we have:
\begin{align}
\mathbb{E}[\Delta^2] &\leq \frac{1}{N} \mathbb{E}_{\widetilde{\mathbf{X}}} [ \| k(\widetilde{\mathbf{X}}, \cdot) \|_{\mathcal{H}_k}^2 ] = \frac{1}{N} \mathbb{E}_{\widetilde{\mathbf{X}}} [ k(\widetilde{\mathbf{X}}, \widetilde{\mathbf{X}}) ] \leq \frac{K}{N} \nonumber \\
\Rightarrow \mathbb{E}[\Delta] &\leq \sqrt{\frac{K}{N}}. \label{eq: jensen}
\end{align}

Combining equations \eqref{eq: mc-darmid} and \eqref{eq: jensen}, we have with probability $1-\delta$:
$$\| \mu_{p*m} - \mu_{\widehat{(p*m)}_N} \|_{\mathcal{H}_k} \leq \sqrt{\frac{K}{N}} + \sqrt{\frac{2K \log(1/\delta)}{N}}$$
Substituting this back into the equation \eqref{eq: decompostion} (multiplying by 2), we have with probability $1-\delta$:
\begin{align*}
\| \mu_{p*m} - \mu_{q_{\widehat{\boldsymbol{\theta}}_N}*m} \|_{\mathcal{H}_k} &\leq \inf_{\boldsymbol{\theta} \in \Theta_J} \| \mu_{p*m} - \mu_{q_{\boldsymbol{\theta}}*m} \|_{\mathcal{H}_k} + 2  \sqrt{\frac{K}{N}} + 2  \sqrt{\frac{2K \log(1/\delta)} {N}} \\
\Rightarrow \operatorname{convMMD}_k\left(p, q_{\widehat{\boldsymbol{\theta}}_N}, m\right) &\leq \inf_{\boldsymbol{\theta} \in \Theta_J} \operatorname{convMMD}_k\left(p, q_{\boldsymbol{\theta}}, m\right)+2 \sqrt{\frac{K}{N}}+2 \sqrt{\frac{2 K \log (1 / \delta)}{N}}
\end{align*}
\end{proof}

\begin{theorem}[Nonparametric $L_2$ Convergence] 
  Under Assumptions~\ref{assump:noise-indep} --  \ref{assump:compactness}, let $q_{\widehat{\boldsymbol{\theta}}_N}$ be the convMMD estimator over $\mathcal{Q}_J$. 

  \textbf{Case A (Ordinary Smooth Noise and Kernel):} Under (OS-N) with exponent $\gamma$ and (OS-K) with exponent $\nu$, choosing $J_N \asymp N^{d/(2\beta)}$ yields
  $$\|q_{\widehat{\boldsymbol{\theta}}_N} - p\|_{2} = O_p\left(N^{-\beta/(2\beta + 2\gamma + \nu)}\right).$$

  \textbf{Case B (Supersmooth Noise or Kernel):} If the noise is supersmooth (SS-N) with exponent $\gamma$, or the kernel is supersmooth (SS-K) with exponent
  $\nu$, or both, define
  $$\gamma^{*} = \begin{cases}
  \gamma & \text{if (SS-N) and (OS-K)}, \\
  \nu & \text{if (OS-N) and (SS-K)}, \\
  \max(\gamma, \nu) & \text{if (SS-N) and (SS-K)}.
  \end{cases}$$
  Then, choosing $J_N \asymp N^{d/(2\beta)}$ yields
  $$\|q_{\widehat{\boldsymbol{\theta}}_N} - p\|_{2} = O_p\left((\log N)^{-\beta/\gamma^{*}}\right).$$
\end{theorem}

\begin{proof}
We aim to bound the expected true $L_2$ risk of the latent density estimator, $\|q_{\widehat{\boldsymbol{\theta}}_N}- p\|_2$. By Plancherel's theorem \citep{rudin1974real}, the squared $L_2$ distance between the true density and our estimator is proportional to the integrated squared difference of their Fourier transforms. Let $\phi_q(\mathbf{t})$ and $\phi_p(\mathbf{t})$ denote the characteristic functions of $q_{\widehat{\boldsymbol{\theta}}_N}$ and $p$ respectively, and let $\Delta \phi(\mathbf{t}) = \phi_q(\mathbf{t}) - \phi_p(\mathbf{t})$. 
We then have:
$$ \|q_{\widehat{\boldsymbol{\theta}}_N} - p\|_2^2 = \frac{1}{(2\pi)^d} \int_{\mathbb{R}^d} |\Delta \phi(\mathbf{t})|^2 d\mathbf{t}. $$

By the convolution theorem and Plancherel's identity (or Bochner's theorem for translation-invariant kernels) \citep{sriperumbudur2010hilbert}, the squared MMD between the convolved densities $q_{\widehat{\boldsymbol{\theta}}_N} *m$ and $p*m$ can be written in the frequency domain as: 
$$ \text{MMD}^2(q_{\widehat{\boldsymbol{\theta}}_N} *m, p*m) = \frac{1}{(2\pi)^d} \int_{\mathbb{R}^d} |\Delta \phi(\mathbf{t})|^2 \underbrace{|\phi_m(\mathbf{t})|^2 \phi_{k}(\mathbf{t})}_{W(\mathbf{t})} d\mathbf{t}, $$
where we recall that $\phi_m(\mathbf{t})$ is the characteristic function of the error distribution $m$, and $\phi_{k}(\mathbf{t})$ is the Fourier transform of the MMD kernel $k$. 
The term $W(\mathbf{t}) = |\phi_m(\mathbf{t})|^2 \phi_{k}(\mathbf{t})$ then acts as a spectral weighting function. Because $W(\mathbf{t}) \to 0$ as $\|\mathbf{t}\| \to \infty$, high-frequency differences are heavily \emph{discounted} (i.e., lightly penalized) in the MMD objective, which makes them difficult to recover and fundamentally characterizes the ill-posed nature of deconvolution.

Under the two-sided smoothness assumptions (Assumptions~\ref{assump:noise-smooth} and \ref{assump:kernel-smooth}), $W(\mathbf{t})$ satisfies the following two-sided bounds:
\begin{itemize}
      \item (OS-N) + (OS-K): $W(\mathbf{t}) {\asymp} (1+\|\mathbf{t}\|^2)^{-\gamma - \nu/2}$
      \item (SS-N) + (OS-K): $W(\mathbf{t}) {\asymp} \exp(-2c_m\|\mathbf{t}\|^\gamma)(1+\|\mathbf{t}\|^2)^{-\nu/2}$
      \item (OS-N) + (SS-K): $W(\mathbf{t}) {\asymp} (1+\|\mathbf{t}\|^2)^{-\gamma} \exp(-c_k\|\mathbf{t}\|^\nu)$
      \item (SS-N) + (SS-K): $W(\mathbf{t}) {\asymp}  \exp(-2c_m\|\mathbf{t}\|^\gamma - c_k\|\mathbf{t}\|^\nu)$
\end{itemize}

To recover the $L_2$ error from the MMD space, we must invert this weight.  We truncate the $L_2$ integral, for any $\Omega > 0$:
\begin{equation} \label{eq: split_integral}
\|q_{\widehat{\boldsymbol{\theta}}_N} - p\|_2^2 = \frac{1}{(2\pi)^d} \left( \underbrace{\int_{\|\mathbf{t}\| \leq \Omega} |\Delta \phi(\mathbf{t})|^2 d\mathbf{t}}_{\text{Low-Frequency Error}} + \underbrace{\int_{\|\mathbf{t}\| > \Omega} |\Delta \phi(\mathbf{t})|^2 d\mathbf{t}}_{\text{High-Frequency Tail}} \right).
\end{equation}
The cutoff $\Omega$ is a free parameter to be optimized at the end; it is not tied to the sieve complexity $J$.

\textbf{Bounding the High-Frequency Tail:}
Because both the true density $p$ and the empirical estimator $q_{\widehat{\boldsymbol{\theta}}_N}$ belong to the Sobolev space $H^\beta$ with bounded norms, their difference $\Delta \phi \in H^\beta$ is also bounded: $\int_{\mathbb{R}^d} (1+\|\mathbf{t}\|^2)^\beta |\Delta \phi(\mathbf{t})|^2 d\mathbf{t} \le C$ for some constant $C$. For any $\|\mathbf{t}\| > \Omega$, the fraction $\frac{(1+\|\mathbf{t}\|^2)^\beta}{(1+\Omega^2)^\beta} \ge 1$. We use this to bound the tail:
\begin{align}
\int_{\|\mathbf{t}\| > \Omega} |\Delta \phi(\mathbf{t})|^2 d\mathbf{t} &\le \int_{\|\mathbf{t}\| > \Omega} \frac{(1+\|\mathbf{t}\|^2)^\beta}{(1+\Omega^2)^\beta} |\Delta \phi(\mathbf{t})|^2 d\mathbf{t} \nonumber \\
&\le \frac{1}{(1+\Omega^2)^\beta} \int_{\mathbb{R}^d} (1+\|\mathbf{t}\|^2)^\beta |\Delta \phi(\mathbf{t})|^2 d\mathbf{t} \le \frac{C}{(1+\Omega^2)^\beta}.\label{eq:high-freq-bound}
\end{align}
For large $\Omega$, the term $(1+\Omega^2)^\beta \asymp \Omega^{2\beta}$. 

\textbf{Bounding the Low-Frequency Error:}
For the low-frequency integral ($\|\mathbf{t}\| \le \Omega$), we multiply and divide by the spectral weight $W(\mathbf{t})$. Using the \emph{lower bound} from the two-sided assumption on $W(\mathbf{t})$, we can control the supremum of the inverse weight: $\sup_{\|\mathbf{t}\| \le \Omega} 1/W(\mathbf{t})$ is finite and grows at a known rate in $\Omega$. By pulling this supremum outside the integral and noting that the integrand is non-negative, extending the bounds back to all of $\mathbb{R}^d$ yields an upper bound matching the true MMD. Note that because we pull the supremum of $1/W(\mathbf{t})$ outside the integral, the geometric volume of the domain $\|\mathbf{t}\| \le \Omega$ is implicitly captured entirely by the full MMD integral, preventing an explicit $\Omega^d$ factor from appearing here:
\begin{align} 
\int_{\|\mathbf{t}\| \leq \Omega} |\Delta \phi(\mathbf{t})|^2 d\mathbf{t} &\le \left( \sup_{\|\mathbf{t}\| \leq \Omega} \frac{1}{W(\mathbf{t})} \right) \int_{\|\mathbf{t}\| \leq \Omega} |\Delta \phi(\mathbf{t})|^2 W(\mathbf{t}) d\mathbf{t} \nonumber \\
&\le \left( \sup_{\|\mathbf{t}\| \leq \Omega} \frac{1}{W(\mathbf{t})} \right) (2\pi)^d \text{MMD}^2(q_{\widehat{\boldsymbol{\theta}}_N} *m, p*m). \label{eq: low_freq_bound}
\end{align}

Squaring both sides of Theorem~\ref{th:oracle} and using $(a+b)^2 \le 2a^2 + 2b^2$:
\begin{equation}
\text{MMD}^2(q_{\widehat{\boldsymbol{\theta}}_N} *m, p*m) \lesssim \inf_{\boldsymbol{\theta} \in \mathcal{Q}_J} \text{MMD}^2(q *m, p*m) + O_p(N^{-1}). \label{eq:mmd-oracle-squared} 
\end{equation} 
Let $q_{J}^{*} \in \mathcal{Q}_J$ be the element achieving the sieve approximation rate $\|q_{J}^{*} - p\|_{2} \le C_{\mathrm{approx}}\,J^{-\beta/d}$ (Assumption~\ref{assump:sieve-approx}).  The spectral weight is bounded globally: $\sup_{\mathbf{t}} W(\mathbf{t}) \le C_W$ (which follows from $\phi_k(\mathbf{t})$ being bounded and $|\phi_m(\mathbf{t})| \le 1$), for some constant $C_W$. By factoring out this maximum weight and applying Plancherel's theorem to the remaining integral, we can relate the MMD distance directly to the $L_2$ distance:
\begin{align}
\operatorname{MMD}_k^2(q_{J}^{*} * m,\; p * m)
&=
\frac{1}{(2\pi)^d}\int_{\mathbb{R}^d} |\phi_{q_{J}^{*}}(\mathbf{t}) - \phi_p(\mathbf{t})|^2\,W(\mathbf{t})\,d\mathbf{t} \nonumber \\
&\le
C_W \left( \frac{1}{(2\pi)^d} \int_{\mathbb{R}^d} |\phi_{q_{J}^{*}}(\mathbf{t}) - \phi_p(\mathbf{t})|^2\,d\mathbf{t} \right) \nonumber \\
&=
C_W\,\|q_{J}^{*} - p\|_{2}^{2} \nonumber \\
&\le
C_W\,C_{\mathrm{approx}}^2\,J^{-2\beta/d}.
\label{eq:mmd-bias}
\end{align}
Substituting into~\eqref{eq:mmd-oracle-squared}:
\begin{equation}
\operatorname{MMD}_k^2(q_{\widehat{\boldsymbol{\theta}}_N} * m,\; p * m)
\lesssim
J^{-2\beta/d} + N^{-1}.
\label{eq:mmd-combined}
\end{equation}

\medskip
Collecting~\eqref{eq:high-freq-bound},~\eqref{eq: low_freq_bound}, and~\eqref{eq:mmd-combined}:
\begin{equation}
\|q_{\widehat{\boldsymbol{\theta}}_N} - p\|_{2}^{2}
\lesssim
\left(\sup_{\|\mathbf{t}\| \leq \Omega} \frac{1}{W(\mathbf{t})}\right) \left(J^{-2\beta/d} + N^{-1}\right) + \Omega^{-2\beta}.
\label{eq:total-risk}
\end{equation}
This bound has two free parameters: $J$ and $\Omega$. We optimize them sequentially.

\emph{Choose $J$:} Select $J$ large enough that the sieve bias is dominated by the statistical error:
\begin{equation}
J^{-2\beta/d} \le N^{-1}
\qquad\Longleftrightarrow\qquad
J \gtrsim N^{d/(2\beta)}.
\end{equation}
This simplifies~\eqref{eq:total-risk} to
\begin{equation}
\|q_{\widehat{\boldsymbol{\theta}}_N} - p\|_{2}^{2}
\lesssim
W^{-1}(\Omega)\,N^{-1} + \Omega^{-2\beta}.
\end{equation}
\textbf{Case A (OS-N) and (OS-K):}
Under polynomial decay assumptions, 
$|\phi_m(\mathbf{t})| \asymp (1+\|\mathbf{t}\|^2)^{-\gamma/2}$ and $\phi_{k}(\mathbf{t}) \asymp (1+\|\mathbf{t}\|^2)^{-\nu/2}$, we have 
\begin{equation}
\sup_{\|\mathbf{t}\| \leq \Omega} \frac{1}{W(\mathbf{t})}
=
\sup_{\|\mathbf{t}\| \leq \Omega} \frac{1}{|\phi_m(\mathbf{t})|^2\,\phi_k(\mathbf{t})} \asymp \{(1+\Omega^{2})^{\gamma/2}\}^{2} \cdot \{(1+\Omega^{2})^{\nu/2}\}
\asymp
\Omega^{2\gamma+\nu}. 
\end{equation}
This simplifies~\eqref{eq:total-risk} to
\begin{equation}
\|q_{\widehat{\boldsymbol{\theta}}_N} - p\|_{2}^{2}
\lesssim
\Omega^{2\gamma+\nu}N^{-1} + \Omega^{-2\beta}.
\end{equation}
Balancing the bias and variance terms yields: 
\begin{equation}
\Omega^{2\gamma+\nu}\,N^{-1} = \Omega^{-2\beta}
\qquad\Longrightarrow\qquad
\Omega_N = N^{1/(2\beta+2\gamma+\nu)}.
\end{equation}
Substituting $\Omega \asymp N^{1/(2\beta+2\gamma+\nu)}$, and taking the square root provides the final polynomial convergence rate:
$$ \|q_{\widehat{\boldsymbol{\theta}}_N} - p\|_{{2}} = O_p\left( \sqrt{ \left( N^{\frac{1}{2\beta + 2\gamma + \nu}} \right)^{-2\beta}} \right) = O_p\left( N^{-\frac{\beta}{2\beta + 2\gamma + \nu}} \right). $$

\textbf{Case B (SS-N) or (SS-K):}
Under exponential decay assumptions, $|\phi_m(\mathbf{t})|^2 \asymp \exp(-2c_m\|\mathbf{t}\|^\gamma)$  or $\phi_{k}(\mathbf{t}) \asymp \exp(-2c_k\|\mathbf{t}\|^\nu)$.
Let
  $$\gamma^{*} = \begin{cases}
  \gamma & \text{if (SS-N) and (OS-K)}, \\
  \nu & \text{if (OS-N) and (SS-K)}, \\
  \max(\gamma, \nu) & \text{if (SS-N) and (SS-K)}.
  \end{cases}$$
Then, using the \emph{lower bound} from the two-sided assumptions and identifying the dominant exponent $\gamma^{*}$, we have:
$$\sup_{\|\mathbf{t}\| \leq \Omega} \frac{1}{W(\mathbf{t})} = \sup_{\|\mathbf{t}\| \leq \Omega} \frac{1}{|\phi_m(\mathbf{t})|^2\,\phi_k(\mathbf{t})} \asymp \exp(c^{*}\Omega^{\gamma^{*}}).$$ 
The variance term dominates the decay and balancing $\exp(c^{*}\Omega^{\gamma^{*}})N^{-1} \asymp \Omega^{-2\beta}$ gives $\Omega_N \asymp (\log N)^{1/\gamma^{*}}$, and substituting yields
\begin{equation}
\|q_{\widehat{\boldsymbol{\theta}}_N} - p\|_{2} = O_p\!\left((\log N)^{-\beta/\gamma^{*}}\right).
\end{equation}
\end{proof}

\begin{remark}
While the dimension $d$ does not appear explicitly in the convergence exponents, it imposes fundamental implicit constraints on the model parameters. Specifically, for the target density to lie in $H^\beta$, we require $\beta > d/2$ per the Sobolev embedding theorem. Furthermore, to ensure the kernel is bounded and integrable (Assumption~\ref{assump:kernel}), its Fourier decay $\nu$ must satisfy $\nu > d$ in Case A, ensuring that $\int (1+\|\mathbf{t}\|^2)^{-\nu/2} d\mathbf{t} < \infty$. 
Thus, higher dimensions strictly require smoother targets and smoother kernels, effectively slowing down the permissible rates.
\end{remark}

\begin{theorem}[Empirical Bayes Denoising Rates]
Let $\widehat{\mathbf{x}}^*$ be the oracle posterior mean and $\widehat{\mathbf{x}}^{\mathrm{EB}}$ be the empirical Bayes estimator using $q_{\widehat{\boldsymbol{\theta}}_N}$. Under Assumptions~\ref{assump:noise-indep} -- \ref{assump:bounded-noise},  with the additional requirement that the true marginal $\widetilde{p}$ is strictly positive on a compact observation domain $\mathcal{K}$, the expected excess empirical Bayes risk evaluated over $\mathcal{K}$ satisfies:
$$ \int_{\mathcal{K}} \|\widehat{\mathbf{x}}^{\mathrm{EB}}(\widetilde{\mathbf{x}}) - \widehat{\mathbf{x}}^*(\widetilde{\mathbf{x}})\|^2 \widetilde{p}(\widetilde{\mathbf{x}}) d\widetilde{\mathbf{x}} = O_p\!\left(\|q_{\widehat{\boldsymbol{\theta}}_N} - p\|_{2}^2\right). $$
Combined with Theorem~\ref{th:l2rate}:
\begin{itemize}
    \item \textbf{Case A (Ordinary-smooth):} $O_p\Big(N^{-\frac{2\beta}{2\beta + 2\gamma + \nu}}\Big)$.
    \item \textbf{Case B (Super-smooth):} $O_p\Big((\log N)^{-\frac{2\beta}{\gamma^*}}\Big)$.
\end{itemize}
\end{theorem}

\begin{proof}
For the rest of the proof, we denote $q = q_{\widehat{\boldsymbol{\theta}}_N}$ for brevity. First, we prove that the total Mean Squared Error (MSE) decomposes exactly into the irreducible Bayes risk and the excess empirical Bayes (EB) risk. We write the total error as:
\begin{align*}
\|\mathbf{X} - \widehat{\mathbf{x}}^{\mathrm{EB}}(\widetilde{\mathbf{X}})\|^2 
&= \|\mathbf{X} - \widehat{\mathbf{x}}^{*}(\widetilde{\mathbf{X}}) + \widehat{\mathbf{x}}^{*}(\widetilde{\mathbf{X}}) - \widehat{\mathbf{x}}^{\mathrm{EB}}(\widetilde{\mathbf{X}})\|^2 \\
&= \|\mathbf{X} - \widehat{\mathbf{x}}^{*}(\widetilde{\mathbf{X}})\|^2 + \|\widehat{\mathbf{x}}^{*}(\widetilde{\mathbf{X}}) - \widehat{\mathbf{x}}^{\mathrm{EB}}(\widetilde{\mathbf{X}})\|^2 + 2\langle \mathbf{X} - \widehat{\mathbf{x}}^{*}(\widetilde{\mathbf{X}}),\, \widehat{\mathbf{x}}^{*}(\widetilde{\mathbf{X}}) - \widehat{\mathbf{x}}^{\mathrm{EB}}(\widetilde{\mathbf{X}}) \rangle.
\end{align*}
Taking the expectation over the joint distribution of $(\mathbf{X}, \widetilde{\mathbf{X}})$, we can apply the tower property of expectation to the cross-term by conditioning on $\widetilde{\mathbf{X}}$:
\begin{align*}
& \mathbb{E}_{\mathbf{X}, \widetilde{\mathbf{X}}} \big[ \langle \mathbf{X} - \widehat{\mathbf{x}}^{*}(\widetilde{\mathbf{X}}),\, \widehat{\mathbf{x}}^{*}(\widetilde{\mathbf{X}}) - \widehat{\mathbf{x}}^{\mathrm{EB}}(\widetilde{\mathbf{X}}) \rangle \big] 
= \mathbb{E}_{\widetilde{\mathbf{X}}} \left[ \mathbb{E}_{\mathbf{X} \mid \widetilde{\mathbf{X}}} \big[ \langle \mathbf{X} - \widehat{\mathbf{x}}^{*}(\widetilde{\mathbf{X}}),\, \widehat{\mathbf{x}}^{*}(\widetilde{\mathbf{X}}) - \widehat{\mathbf{x}}^{\mathrm{EB}}(\widetilde{\mathbf{X}}) \rangle \mid \widetilde{\mathbf{X}} \big] \right] \\
&= \mathbb{E}_{\widetilde{\mathbf{X}}} \left[ \langle \underbrace{\mathbb{E}_{\mathbf{X} \mid \widetilde{\mathbf{X}}}[\mathbf{X} \mid \widetilde{\mathbf{X}}] - \widehat{\mathbf{x}}^{*}(\widetilde{\mathbf{X}})}_{=\, \mathbf{0}},\, \widehat{\mathbf{x}}^{*}(\widetilde{\mathbf{X}}) - \widehat{\mathbf{x}}^{\mathrm{EB}}(\widetilde{\mathbf{X}}) \rangle \right] = \mathbf{0},
\end{align*}
where the term vanishes because the oracle estimator $\widehat{\mathbf{x}}^{*}(\widetilde{\mathbf{X}})$ is defined exactly as the true posterior conditional mean $\mathbb{E}_p[\mathbf{X} \mid \widetilde{\mathbf{X}}]$. Thus, the total MSE decomposes perfectly:
\begin{equation}
    \mathbb{E}\big[\|\mathbf{X} - \widehat{\mathbf{x}}^{\mathrm{EB}}(\widetilde{\mathbf{X}})\|^2\big] = \underbrace{\mathbb{E}\big[\|\mathbf{X} - \widehat{\mathbf{x}}^{*}(\widetilde{\mathbf{X}})\|^2\big]}_{\text{Irreducible Bayes Risk}} + \underbrace{\mathbb{E}\big[\|\widehat{\mathbf{x}}^{\mathrm{EB}}(\widetilde{\mathbf{X}}) - \widehat{\mathbf{x}}^{*}(\widetilde{\mathbf{X}})\|^2\big]}_{\text{Excess EB Risk}}.
\end{equation}
Our goal is now to bound the excess EB risk.

By Bayes' rule, the true posterior density is $\pi^{*}(\mathbf{x} \mid \widetilde{\mathbf{x}}) = \frac{p(\mathbf{x}) m(\widetilde{\mathbf{x}}-\mathbf{x})}{\widetilde{p}(\widetilde{\mathbf{x}})}$. The oracle posterior mean, as a function of the noisy observation $\widetilde{\mathbf{x}}$, is:
\begin{equation}
    \widehat{\mathbf{x}}^{\widetilde{\mathbf{x}}} = \int \mathbf{x} \pi^{*}(\mathbf{x} \mid \widetilde{\mathbf{x}}) d\mathbf{x} = \frac{\int \mathbf{x}\,m(\widetilde{\mathbf{x}}-\mathbf{x})\,p(\mathbf{x})\,d\mathbf{x}}{\widetilde{p}(\widetilde{\mathbf{x}})}.
\end{equation}
We define the numerator functional $\mathbf{r}_p(\widetilde{\mathbf{x}}) = \int \mathbf{x}\,m(\widetilde{\mathbf{x}}-\mathbf{x})\,p(\mathbf{x})\,d\mathbf{x}$, making $\widehat{\mathbf{x}}^{\widetilde{\mathbf{x}}} = \frac{\mathbf{r}_p(\widetilde{\mathbf{x}})}{\widetilde{p}(\widetilde{\mathbf{x}})}$. Analogously, the empirical Bayes estimator is $\widehat{\mathbf{x}}^{\mathrm{EB}}(\widetilde{\mathbf{x}}) = \frac{\mathbf{r}_q(\widetilde{\mathbf{x}})}{\widetilde{q}(\widetilde{\mathbf{x}})}$.

We decompose the pointwise difference between the estimators by adding and subtracting $\frac{\mathbf{r}_p(\widetilde{\mathbf{x}})}{\widetilde{q}(\widetilde{\mathbf{x}})}$:
\begin{align}
\widehat{\mathbf{x}}^{\mathrm{EB}}(\widetilde{\mathbf{x}}) - \widehat{\mathbf{x}}^{*}(\widetilde{\mathbf{x}}) 
&= \frac{\mathbf{r}_q(\widetilde{\mathbf{x}})}{\widetilde{q}(\widetilde{\mathbf{x}})} - \frac{\mathbf{r}_p(\widetilde{\mathbf{x}})}{\widetilde{p}(\widetilde{\mathbf{x}})} \nonumber \\
&= \frac{\mathbf{r}_q(\widetilde{\mathbf{x}})}{\widetilde{q}(\widetilde{\mathbf{x}})} - \frac{\mathbf{r}_p(\widetilde{\mathbf{x}})}{\widetilde{q}(\widetilde{\mathbf{x}})} + \frac{\mathbf{r}_p(\widetilde{\mathbf{x}})}{\widetilde{q}(\widetilde{\mathbf{x}})} - \frac{\mathbf{r}_p(\widetilde{\mathbf{x}})}{\widetilde{p}(\widetilde{\mathbf{x}})} \nonumber \\
&= \frac{\mathbf{r}_q(\widetilde{\mathbf{x}}) - \mathbf{r}_p(\widetilde{\mathbf{x}})}{\widetilde{q}(\widetilde{\mathbf{x}})} + \mathbf{r}_p(\widetilde{\mathbf{x}}) \left( \frac{\widetilde{p}(\widetilde{\mathbf{x}}) - \widetilde{q}(\widetilde{\mathbf{x}})}{\widetilde{p}(\widetilde{\mathbf{x}})\widetilde{q}(\widetilde{\mathbf{x}})} \right) \nonumber \\
&= \underbrace{\frac{\mathbf{r}_q(\widetilde{\mathbf{x}}) - \mathbf{r}_p(\widetilde{\mathbf{x}})}{\widetilde{q}(\widetilde{\mathbf{x}})}}_{(A)} + \underbrace{\widehat{\mathbf{x}}^{*}(\widetilde{\mathbf{x}}) \cdot \frac{\widetilde{p}(\widetilde{\mathbf{x}}) - \widetilde{q}(\widetilde{\mathbf{x}})}{\widetilde{q}(\widetilde{\mathbf{x}})}}_{(B)}. \label{eq:AB-decomp-proof}
\end{align}

We bound the numerators of terms $(A)$ and $(B)$ using the assumption that $p$ and $q$ are supported on a compact domain $\mathcal{X}$ where $\|\mathbf{x}\| \le M$. Using the triangle inequality for integrals and the fact that $m(\cdot) \ge 0$:
\begin{align}
\|\mathbf{r}_q(\widetilde{\mathbf{x}}) - \mathbf{r}_p(\widetilde{\mathbf{x}})\| 
&= \left\| \int \mathbf{x}\, m(\widetilde{\mathbf{x}}-\mathbf{x}) (q(\mathbf{x}) - p(\mathbf{x})) d\mathbf{x} \right\| \nonumber \\
&\le \int \|\mathbf{x}\|\, m(\widetilde{\mathbf{x}}-\mathbf{x}) |q(\mathbf{x}) - p(\mathbf{x})| d\mathbf{x} \nonumber \\
&\le M \int m(\widetilde{\mathbf{x}}-\mathbf{x}) |q(\mathbf{x}) - p(\mathbf{x})| d\mathbf{x} \nonumber \\
&= M\,(|q-p|*m)(\widetilde{\mathbf{x}}).
\end{align}
Similarly, the difference in the marginal densities is bounded by:
\begin{align}
|\widetilde{p}(\widetilde{\mathbf{x}}) - \widetilde{q}(\widetilde{\mathbf{x}})| 
&= \left| \int m(\widetilde{\mathbf{x}}-\mathbf{x}) (p(\mathbf{x}) - q(\mathbf{x})) d\mathbf{x} \right| \nonumber \\
&\le \int m(\widetilde{\mathbf{x}}-\mathbf{x}) |p(\mathbf{x}) - q(\mathbf{x})| d\mathbf{x} \nonumber \\
&= (|q-p|*m)(\widetilde{\mathbf{x}}).
\end{align}
Because $\widehat{\mathbf{x}}^{*}(\widetilde{\mathbf{x}})$ is the expectation of $\mathbf{X}$ under the true posterior (a distribution strictly supported on $\mathcal{X}$ where $\|\mathbf{x}\| \le M$), its magnitude is bounded by the maximum norm of the support: $\|\widehat{\mathbf{x}}^{*}(\widetilde{\mathbf{x}})\| \le M$.

Applying the triangle inequality to Equation~\eqref{eq:AB-decomp-proof}:
\begin{align}
\|\widehat{\mathbf{x}}^{\mathrm{EB}}(\widetilde{\mathbf{x}}) - \widehat{\mathbf{x}}^{*}(\widetilde{\mathbf{x}})\|
&\le \frac{\|\mathbf{r}_q(\widetilde{\mathbf{x}}) - \mathbf{r}_p(\widetilde{\mathbf{x}})\|}{\widetilde{q}(\widetilde{\mathbf{x}})} + \|\widehat{\mathbf{x}}^{*}(\widetilde{\mathbf{x}})\| \frac{|\widetilde{p}(\widetilde{\mathbf{x}}) - \widetilde{q}(\widetilde{\mathbf{x}})|}{\widetilde{q}(\widetilde{\mathbf{x}})} \nonumber \\
&\le \frac{1}{\widetilde{q}(\widetilde{\mathbf{x}})} \Big[ M\,(|q-p|*m)(\widetilde{\mathbf{x}}) + M\,(|q-p|*m)(\widetilde{\mathbf{x}}) \Big] \nonumber \\
&= \frac{2M}{\widetilde{q}(\widetilde{\mathbf{x}})}(|q-p|*m)(\widetilde{\mathbf{x}}).
\end{align}

We now bound the denominator $\widetilde{q}(\widetilde{\mathbf{x}})$ using uniform convergence of convolved densities.

% Let $\mathcal{K} \subset \mathbb{R}^d$ be a compact observation domain. Because the true marginal density is strictly positive on $\mathcal{K}$, there exists an absolute constant $\rho_0 > 0$ such that $\inf_{\widetilde{\mathbf{x}} \in \mathcal{K}} \widetilde{p}(\widetilde{\mathbf{x}}) \ge 2\rho_0$. Since $\widetilde{q}$ converges uniformly to $\widetilde{p}$ as $N \to \infty$, we have $\widetilde{q}(\widetilde{\mathbf{x}}) \ge \rho_0$ on $\mathcal{K}$ with high probability. 

Let $\mathcal{K} \subset \mathbb{R}^d$ be a compact observation domain. By the Cauchy-Schwarz inequality, the uniform ($L_\infty$) distance between the marginal densities is explicitly controlled by the $L_2$ error of the latent densities:
\begin{align}
\|\widetilde{q} - \widetilde{p}\|_\infty 
&= \sup_{\widetilde{\mathbf{x}}} \left| \int (q(\mathbf{x}) - p(\mathbf{x})) m(\widetilde{\mathbf{x}} - \mathbf{x}) d\mathbf{x} \right| \nonumber \\
&\le \sup_{\widetilde{\mathbf{x}}} \left( \int |q(\mathbf{x}) - p(\mathbf{x})|^2 d\mathbf{x} \right)^{1/2} \left( \int m(\widetilde{\mathbf{x}} - \mathbf{x})^2 d\mathbf{x} \right)^{1/2} \nonumber \\
&= \|q - p\|_{2} \|m\|_{2}.
\end{align}
By assumption, the noise is a bounded probability density ($\|m\|_\infty < \infty$). Therefore, it is strictly square-integrable, since $\|m\|_{2}^2 \le \|m\|_\infty \|m\|_{1} = \|m\|_\infty < \infty$. Because Theorem~\ref{th:l2rate} establishes that the latent error $\|q_{\widehat{\boldsymbol{\theta}}_N} - p\|_{2} \to 0$ in probability as $N \to \infty$, the bound above mathematically guarantees that the marginal density $\widetilde{q}$ converges uniformly to $\widetilde{p}$ in probability.

Because the true marginal density is strictly positive on $\mathcal{K}$, there exists an absolute constant $\rho_0 > 0$ such that $\inf_{\widetilde{\mathbf{x}} \in \mathcal{K}} \widetilde{p}(\widetilde{\mathbf{x}}) \ge 2\rho_0$. The uniform convergence guarantees that, with probability approaching one for sufficiently large $N$, the estimator satisfies $\widetilde{q}(\widetilde{\mathbf{x}}) \ge \rho_0$ simultaneously for all $\widetilde{\mathbf{x}} \in \mathcal{K}$. 

Applying the triangle inequality to the pointwise difference over $\mathcal{K}$:
\begin{equation}
\|\widehat{\mathbf{x}}^{\mathrm{EB}}(\widetilde{\mathbf{x}}) - \widehat{\mathbf{x}}^*(\widetilde{\mathbf{x}})\| 
\le \frac{\|r_q - r_p\|}{\widetilde{q}} + M \frac{|\widetilde{p} - \widetilde{q}|}{\widetilde{q}} 
\le \frac{2M}{\rho_0}\,(|q-p|*m)(\widetilde{\mathbf{x}}).
\end{equation}
 We square the pointwise bound and integrate over the observation domain $\mathcal{K}$ with respect to the true marginal density $\widetilde{p}$:
\begin{align}
\int_{\mathcal{K}} \|\widehat{\mathbf{x}}^{\mathrm{EB}}(\widetilde{\mathbf{x}}) - \widehat{\mathbf{x}}^*(\widetilde{\mathbf{x}})\|^2 \widetilde{p}(\widetilde{\mathbf{x}}) d\widetilde{\mathbf{x}}
&\le \frac{4M^2}{\rho_0^2} \int_{\mathcal{K}} \big((|q-p|*m)(\widetilde{\mathbf{x}})\big)^2 \widetilde{p}(\widetilde{\mathbf{x}}) d\widetilde{\mathbf{x}}.
\end{align}

To bound the $\widetilde{p}(\widetilde{\mathbf{x}})$ term, we utilize the smoothness of the latent space. By Assumption~\ref{assump:latent-sobolev}, the true latent density $p$ belongs to the fractional Sobolev space $H^\beta(\mathbb{R}^d)$ with smoothness index $\beta > d/2$. Under this condition, the Sobolev Embedding Theorem \citep{adams2003sobolev} guarantees that $H^\beta(\mathbb{R}^d)$ embeds continuously into the space of bounded, continuous functions $L^\infty(\mathbb{R}^d) \cap C(\mathbb{R}^d)$. Consequently, the true latent density has a finite supremum norm: $\|p\|_\infty < \infty$. 

The true marginal density is defined by the convolution $\widetilde{p} = p * m$. We apply Young's convolution inequality \citep{folland1999real}, which states that for any $f \in L^\infty$ and $g \in L_1$, their convolution is globally bounded by $\|f * g\|_\infty \le \|f\|_\infty \|g\|_{1}$. Since the noise distribution $m$ is a valid probability density, its total absolute mass is exactly one ($\|m\|_{1} = \int m(\mathbf{u})d\mathbf{u} = 1$). Therefore, the marginal density is globally bounded across all of $\mathbb{R}^d$ by the peak of the latent density:
\begin{equation}
\widetilde{p}(\widetilde{\mathbf{x}}) \le \|p * m\|_\infty \le \|p\|_\infty \|m\|_{1} = \|p\|_\infty.
\end{equation}
Substituting this uniform upper bound into our integral, and extending the integration domain from $\mathcal{K}$ to the entirety of $\mathbb{R}^d$ (which is strictly valid since the squared integrand is non-negative), we obtain:
\begin{align}
\mathbb{E}_{\widetilde{p}}\!\left[\|\widehat{\mathbf{x}}^{\mathrm{EB}} - \widehat{\mathbf{x}}^*\|^2  \right] 
&\le \frac{4M^2 \|p\|_\infty}{\rho_0^2} \int_{\mathbb{R}^d} \big((|q-p|*m)(\widetilde{\mathbf{x}})\big)^2 d\widetilde{\mathbf{x}} \nonumber \\
&= \frac{4M^2 \|p\|_\infty}{\rho_0^2} \big\||q-p|*m\big\|_{2}^2.
\end{align}

By Young's convolution inequality, $\|f * g\|_2 \le \|f\|_2 \|g\|_1$. Since $\|m\|_1 = 1$:
\begin{equation}
\|(q-p)*m\|_{2} \le \|q-p\|_{2}.
\end{equation}

Thus, the excess empirical Bayes risk over $\mathcal{K}$ is directly bounded by the $L_2$ error of the latent density:
\begin{equation}
\mathbb{E}_{\widetilde{p}}\!\left[\|\widehat{\mathbf{x}}^{\mathrm{EB}} - \widehat{\mathbf{x}}^*\|^2\right] \le \frac{4M^2 \|p\|_\infty}{\rho_0^2} \|q_{\widehat{\boldsymbol{\theta}}_N}-p\|_{2}^{2}.
\end{equation}
Because $M$, $\|p\|_\infty$, and $\rho_0$ are absolute constants independent of $N$, the asymptotic rate is entirely governed by the squared $L_2$ norm. 

Squaring the rates established in Theorem~\ref{th:l2rate}:
\begin{itemize}
    \item \textbf{Ordinary-smooth regime:} $\|q_{\widehat{\boldsymbol{\theta}}_N} - p\|_{2} = O_p\Big(N^{-\frac{\beta}{2\beta + 2\gamma + \nu}}\Big) \implies \text{Excess Risk} = O_p\Big(N^{-\frac{2\beta}{2\beta + 2\gamma + \nu}}\Big)$.
    \item \textbf{Super-smooth regime:} $\|q_{\widehat{\boldsymbol{\theta}}_N} - p\|_{2} = O_p\Big((\log N)^{-\frac{\beta}{\gamma^{*}}}\Big) \implies \text{Excess Risk} = O_p\Big((\log N)^{-\frac{2\beta}{\gamma^{*}}}\Big)$.
\end{itemize}
This completes the proof.
\end{proof}

\section{Ordinary- and Super-smooth Regimes: Definitions and Examples}
\label{app:smooth-examples}

This section provides canonical examples underlying the two regimes in Theorem~\ref{th:l2rate}. These regimes are characterized by the decay of the characteristic function of the noise distribution, which determines the degree of ill-posedness of the deconvolution problem.

\paragraph{Setup.}
Let $\varepsilon \sim m$ denote the additive noise with characteristic function
\[
\phi_m(\mathbf{t}) = \mathbb{E}[e^{i \mathbf{t}^\top \varepsilon}].
\]
The behavior of $|\phi_m(\mathbf{t})|$ as $\|\mathbf{t}\| \to \infty$ governs how much high-frequency information about the latent density is lost.

\paragraph{Ordinary-smooth noise.}
The noise distribution is called \emph{ordinary-smooth} with exponent $\gamma > 0$ if
\[
|\phi_m(\mathbf{t})| \asymp (1 + \|\mathbf{t}\|^2)^{-\gamma/2}.
\]
In this case, the characteristic function decays polynomially, and the inverse problem is mildly ill-posed.

A canonical example is \textbf{Laplace noise}. If $\varepsilon \sim \mathrm{Laplace}(\mathbf{0}, b)$, then
\[
\phi_m(\mathbf{t}) = \frac{1}{1 + b^2 \|\mathbf{t}\|^2},
\]
which corresponds to $\gamma = 2$. More generally, Laplace-type and certain heavy-tailed distributions fall into this class.

Under ordinary-smooth noise and kernel assumptions, Theorem~\ref{th:l2rate} yields the polynomial convergence rate
\[
\|q_{\widehat{\boldsymbol{\theta}}_N} - p\|_2
=
O_p\!\left(
N^{-\frac{\beta}{2\beta + 2\gamma + \nu}}
\right).
\]

\paragraph{Super-smooth noise.}
The noise distribution is called \emph{super-smooth} with exponent $\gamma > 0$ if
\[
|\phi_m(\mathbf{t})| \asymp \exp(-c \|\mathbf{t}\|^\gamma)
\]
for some $c > 0$. In this case, the characteristic function decays exponentially, and the inverse problem is severely ill-posed.

A canonical example is \textbf{Gaussian noise}. If $\varepsilon \sim \mathcal{N}(\mathbf{0}, \sigma^2 \mathbf{I}_d)$, then
\[
\phi_m(\mathbf{t}) = \exp\!\left(-\frac{\sigma^2 \|\mathbf{t}\|^2}{2}\right),
\]
corresponding to $\gamma = 2$. 
In this regime, Theorem~\ref{th:l2rate} gives the logarithmic rate
\[
\|q_{\widehat{\boldsymbol{\theta}}_N} - p\|_2
=
O_p\!\left((\log N)^{-\beta/\gamma^*}\right),
\]
where $\gamma^*$ captures the dominant smoothness between the noise and kernel.

\paragraph{Interpretation.}
These two regimes reflect how aggressively the noise suppresses high-frequency components of the latent density. In the ordinary-smooth case, information is lost gradually, allowing polynomial recovery rates. In the super-smooth case, high-frequency components are exponentially attenuated, making recovery substantially harder and limiting convergence to logarithmic rates. This distinction is classical in nonparametric deconvolution and explains why Gaussian measurement error is significantly more challenging than Laplace noise.

\section{Experiment Details} \label{sec: exp-details}

\subsection{Verifying Rates of Convergence}

\paragraph{Data Generating Process (DGP)}
The true latent distribution $p$ of the scalar variable $X$ is a two-component Laplace mixture:
$$ X_i \sim 0.6 \, \text{Laplace}(\mu=3.0, \sigma=0.8) + 0.4 \, \text{Laplace}(\mu=-2.0, \sigma=1.0) $$
where $\sigma$ denotes the standard deviation (corresponding to a scale parameter $b = \sigma/\sqrt{2}$). 

For the observational noise, we simulate a heteroscedastic setting (consistent with Assumption~\ref{assump:known-noise}). For each sample $i$, a noise standard deviation is drawn uniformly: $\phi_i \sim \mathcal{U}(0.5, 1.0)$. We observe the noisy proxy $\widetilde{X}_i = X_i + U_{X, i}$, evaluating two specific noise types:
\begin{itemize}
    \item \textbf{Ordinary-Smooth (OS):} $U_{X, i} \sim \text{Laplace}(0, b = \phi_i/\sqrt{2})$.
    \item \textbf{Super-Smooth (SS):} $U_{X, i} \sim \mathcal{N}(0, \phi_i^2)$.
\end{itemize}

\paragraph{Theoretical Slope Calculations}
The theoretical slopes plotted in Figure~\ref{fig:rkhs-rate} are derived directly from Theorem~\ref{th:l2rate}. 
\begin{itemize}
    \item OS Case (Laplace): For Laplace distribution, we have polynomial convergence exponent $\frac{-2\beta}{2\beta + 2\gamma + \nu}$ and this theoretical limit evaluates to a slope of $-0.33$ in log-log space. Here, $\beta = 1.5$, $\gamma = 2$ and $\nu = 2$.
    \item SS Case (Gaussian): For Gaussian distribution, the convergence rate degrades to $(\log N)^{-\beta/\gamma^{*}}$. This corresponds to a slope of $-1.5$ in $\log$-$\log\log$ space. Here, $\beta = 1.5$, $\gamma = 2$ and $\nu = 2$.
    \item  Latent MMD (Corollary~\ref{cor:latent}): The squared MMD parametric rate is $O_p(N^{-1})$, yielding a log-log slope of $-1.0$.
\end{itemize}
\paragraph{Sieve Architectures}
We utilized two distinct sieve architectures to demonstrate that our bounds are agnostic to the choice of the underlying well-specified sieve:
\begin{itemize}
    \item  Gaussian Mixture Model (GMM):  The models are initialized via standard EM (scikit-learn) on the noisy data before being fine-tuned via the convMMD loss. We scale the number of components $K \in [12, 16]$ gracefully with $N$.  To ensure stability and differentiability, the covariance matrices are parameterized via their Cholesky factors $L$, where the diagonals are constrained using a Softplus activation. For optimization, we use the reparameterization trick to draw Monte Carlo samples (128 samples per batch) to evaluate the expected convMMD objective.
    \item  Normalizing Flow Model (NFM): We utilize a 1D Rational Quadratic Spline (RQS) flow. To strictly satisfy the finite Sobolev norm assumptions required by our theory, we bound the derivative of the splines using a Sigmoid activation (maximum derivative $= 10.0$). The base distribution is a standard Normal matching the empirical mean and variance of the observations. The flow consists of 4 blocks, 16 spline bins, and a tail bound of $30.0$. The hidden sizes of the conditioner MLPs were scaled between 25 and 128 depending on $N$ to guarantee sufficient sieve capacity.
\end{itemize}
\paragraph{Optimization}
The models are trained to minimize the empirical unbiased U-statistic MMD loss. We use AdamW optimizer \cite{loshchilovdecoupled} with a weight decay of $10^{-4}$ (GMM) or $5\times 10^{-3}$ (NFM). Gradients are clipped to a global norm of $1.0$. For learning rate, we use cosine decay schedule starting at $10^{-3}$ (NFM) or $5 \times 10^{-2}$ (GMM), decaying down to a factor of $10^{-4}$ over the training run. Each $N$ configuration is run over 20 to 100 independent random seeds. Training spans $1000$ to $1500$ epochs using batch sizes ranging from $200$ to $400$.

\paragraph{Results}
Figure~\ref{fig:rkhs-rate} evaluates the theoretical predictions of Corollary~\ref{cor:latent} and Theorem~\ref{th:l2rate}. We consider a one-dimensional density deconvolution problem in which the latent variable follows a two-component Laplace mixture distribution and the observations are generated as
\[
\tilde{X}_i = X_i + \varepsilon_i,
\]
with heteroscedastic additive noise. The convMMD estimator is trained by minimizing the discrepancy between the empirical noisy distribution and the noise-convolved model distribution.

We compare two expressive sieve classes: Gaussian mixture models and normalizing flows parameterized by rational quadratic splines. To probe the two inverse-problem regimes in Theorem~\ref{th:l2rate}, we use Laplace noise with a Laplace kernel for the ordinary-smooth setting and Gaussian noise with a Gaussian kernel for the supersmooth setting. Sample sizes are varied over
\[
N \in \{4000, 8000, 12000, 16000, 24000, 32000\}.
\]

We measure nonparametric recovery using
\[
\mathrm{MISE} \approx \|q_{\widehat{\theta}_N} - p\|_{L_2}^2.
\]
In the ordinary-smooth case, the theory predicts polynomial decay; for the Laplace smoothness parameters used here, the expected log-log slope is approximately $-0.33$. Both GMM and normalizing-flow sieves exhibit the predicted linear trend. In the supersmooth case, the problem is severely ill-posed and the theory predicts logarithmic convergence. The empirical $\log(\mathrm{MISE})$ versus $\log(\log N)$ trends align closely with the predicted slope of $-1.5$.

We also evaluate latent-space convergence through
\[
\mathrm{MMD}^2(p, q_{\widehat{\theta}_N}).
\]
Consistent with Corollary~\ref{cor:latent}, the squared latent MMD decays at the parametric rate $O_p(N^{-1})$ in both ordinary-smooth and supersmooth regimes. Thus, even when $L_2$ recovery is slowed by supersmooth noise, the RKHS metric retains fast parametric convergence. 

\subsection{Denoising in presence of outliers} \label{sec:2d_denoising_details}
\paragraph{Data Generating Process and Outlier Corruption}
For each of the three topologies (Two Moons, Circle, Checkerboard), we generate $N=2000$ latent samples $\mathbf{x}_i \in \mathbb{R}^2$. To evaluate robustness, we randomly select $3\%$ of the latent samples and displace them using a large uniform offset (e.g., $\mathcal{U}(-4.0, 4.0)$). This creates a set of "true" outliers that exist far outside the primary manifold.

We generate independent per-sample noise standard deviations $\phi_i \sim \mathcal{U}(a, b)$. For the Moon and Checkerboard datasets, we use noise bounds $(a, b) = (0.2, 0.6)$. Because the Circle has a larger geometric footprint, we increase the noise bounds to $(a,b) = (0.5, 1.0)$. The final observed data is generated as $\widetilde{\mathbf{x}}_i = \mathbf{x}_i + \mathbf{u}_{x,i}$, where $\mathbf{u}_{x,i} \sim \mathcal{N}(\mathbf{0}, \phi_i^2 \mathbf{I}_2)$.

\paragraph{Method Configuration}

\subparagraph{convMMD\_NF}
We utilize a Rational Quadratic Spline (RQS) Normalizing Flow as the sieve model \citep{durkan2019neural}. The flow consists of 4 coupling blocks, with 16 spline bins per block, and conditioner MLPs with a hidden dimension of 32. The tail bound is set to $2.0$. The base distribution is a standard Gaussian shifted and scaled to match the empirical mean and variance of the observations.
The model is trained for 3000 epochs. We use the AdamW optimizer with a learning rate of $5 \times 10^{-2}$ and a cosine decay schedule. Kernel bandwidths are chosen automatically using the median heuristic \citep{schrab2025practical,biggs2023mmd}. We use multiple bandwidths and select bandwidth using the formula $h = \sqrt{Q_q} / (\log_{10}(2n) \sqrt{2})$, where $Q_q$ are the quantiles of the pairwise distances between observed noisy samples and model generated noisy samples. Here, $q=[0.1, 0.2, 0.3, 0.4, 0.5, 0.6, 0.7, 0.8, 0.9]$. Once trained, we draw $M = 30,000$ samples from the learned prior $q_{\widehat{\theta}_N}$ and compute the posterior mean for each observation via Monte Carlo integration:
$$ \widehat{\mathbf{x}}_i^{\mathrm{EB}} = \mathbb{E}_{q_{\widehat{\boldsymbol{\theta}}_N}}[\mathbf{x} \mid \widetilde{\mathbf{x}}_i] \approx \sum_{m=1}^M w_{i,m} \mathbf{y}_m, \quad \text{where} \quad w_{i,m} \propto \exp\left(-\frac{\|\widetilde{\mathbf{x}}_i - \mathbf{y}_m\|^2}{2\phi_i^2}\right), $$

\subparagraph{NPEB}
We use the Nonparametric Empirical Bayes framework via the `npeb` package. This method computes the Nonparametric Maximum Likelihood Estimator (NPMLE) using the `GLMixture` solver. The algorithm places a fine grid over the observational space and optimizes the mixture weights using convex optimization (via CVXPY). We use diagonal precision matrices parameterized by the known heteroscedastic noise variances. The denoised estimates are then extracted using the exact discrete posterior mean (GMLEB).

\subparagraph{XDGMM}
Extreme Deconvolution (XDGMM) fits a continuous Gaussian Mixture Model to data with known heteroscedastic measurement noise. To ensure a fair comparison and avoid bottlenecks, we implemented a custom JAX-accelerated, JIT-compiled version of the XDGMM Expectation-Maximization algorithm. We set the number of Gaussian components to $K=10$ and regularize the covariances by $10^{-6}$ to prevent singularity collapses caused by the outliers. The model is trained for 100 EM iterations (or until convergence with a tolerance of $10^{-5}$). Once fit, we calculate the exact closed-form Gaussian posterior mean for each observation. 

\paragraph{Evaluation}
Performance is measured using the standard Mean Squared Error (MSE) between the denoised estimates $\widehat{\mathbf{x}}_i^{\mathrm{EB}}$ and the clean latent truths $\mathbf{x}_i$ (excluding the outliers):
$$ \text{MSE} = \frac{1}{N} \sum_{i=1}^N \|\widehat{\mathbf{x}}_i^{\mathrm{EB}} - \mathbf{x}_i\|^2 $$
The MSE results presented in Figure~\ref{fig:2d-exp} are representative of the average performance across 5 independent random seeds.

\subsubsection{Kernel Bandwidth Sensitivity}\label{sec:bw_sensivity}

In this experiment, we use convMMD with Normalizing Flow. We use a multi-scale Laplace kernel with bandwidths $\{{\sigma_j}\}_{j=1}^9$ selected via the median heuristic: pairwise distances between observed noisy samples and model generated noisy samples are pooled, nine quantiles of their empirical distribution are extracted, and each $\sigma_j$ is set proportionally to the square root of the corresponding quantile, normalised by $\log_{10}(2n)\sqrt{2}$. This heuristic provides user with a principled data-driven default. Here, we verify the sensitivity of the learned prior and downstream denoising performance on bandwidth choice. 

\subparagraph{Setup.} We hold the full training procedure fixed and vary only a single multiplicative scale factor $s \in {0.10, 0.25, 0.50, 0.75, 1.00, 1.50, 2.00, 3.00, 5.00}$ applied uniformly to all nine bandwidths, i.e.\ $\tilde{\sigma}_j = s \cdot \sigma_j^{\text{heuristic}}$. The value $s = 1$ recovers the default heuristic. For each $(s, \text{dataset}, \text{seed})$ triplet we train a fresh normalizing flow from the same initialization using the scaled bandwidths and evaluate on the held-out test set of 2000 points. Base bandwidths are computed once per seed before training of the model. All other hyperparameters (architecture, learning rate, number of gradient steps) match the main experiment. Results are averaged over five seeds and we report mean $\pm$ standard deviation.
As a bandwidth-independent reference, we include the Oracle baseline (KDE fitted on a large clean sample with no outliers), whose MSE and W1 do not depend on $s$ and are shown as horizontal dotted lines in Figure~\ref{fig:bandwidth-sens-1}.
\subparagraph{Results.} Figure~\ref{fig:bandwidth-sens-1} plots test MSE (left) and prior W1 distance (right) as a function of $s$ for all three datasets, with shaded $\pm$1 std bands across seeds. We see that performance is remarkably stable for $s \in [0.25, 2.00]$, covering nearly an order of magnitude around the heuristic default. Both MSE and W1 remain within one standard deviation of the $s = 1$ value across this range, indicating that convMMD is not sensitive to moderate misspecification of the bandwidth. At $s = 0.10$ (very narrow kernels) the loss becomes dominated by near-zero-distance interactions and the prior loses its ability to capture global structure, causing an increase in W1. At $s = 5.00$ (very wide kernels) the kernel smooths over local geometry and W1 increases similarly, though the degradation is more gradual. 
These results justify using the median heuristic as a fixed, tuning-free component of the conMMD pipeline. 

\subsection{Deconvolution Performance across Dimensions}

\label{app:deconv_with_dim}
\paragraph{Dataset: Tangled Ribbon Manifold}

The latent distribution is supported on a two-dimensional spiral manifold
embedded in $\mathbb{R}^D$.  A scalar parameter $t \sim \mathcal{U}(1.5\pi,
4.5\pi)$ parameterises the base spiral via
\begin{equation}
    \mathbf{x}_{\text{base}} = \Bigl(\tfrac{t\cos t}{t_{\max}},\;
    \tfrac{t\sin t}{t_{\max}}\Bigr) \in [-1,1]^2,
\end{equation}
padded with $D-2$ zeros to form a $D$-dimensional vector, and then multiplied
by a fixed random orthogonal matrix $Q \in \mathbb{R}^{D \times D}$ (drawn once per seed via QR decomposition of a standard Gaussian matrix).  This construction
ensures the manifold is genuinely embedded in $\mathbb{R}^D$ without being
axis-aligned, and that the intrinsic dimension remains 2 regardless of $D$. For each run, we generate $N_{\text{total}} = 12\,000$ samples: the first 2\,000 noisy observations $\{\mathbf{y}_i\}$ form the training set, and the remaining 10\,000 clean latent samples $\{\mathbf{x}_i\}$ form the held-out evaluation set used exclusively
for computing SWD.

\paragraph{Noise model.}
Each observation in the high-dimensional experiments is similarly corrupted by heteroscedastic isotropic Gaussian noise:
\begin{equation}
    \widetilde{\mathbf{x}}_i = \mathbf{x}_i + \mathbf{u}_{x,i}, \qquad
    \mathbf{u}_{x,i} \sim \mathcal{N}(\mathbf{0},\, \phi_i^2 \mathbf{I}_d),
    \qquad
    \phi_i \sim \mathcal{U}(0.085,\, 0.127),
\end{equation}
giving a population mean of $\bar{\phi} = 0.106$. 

Each method is run under three noise-knowledge conditions, exercising both the homoscedastic/heteroscedastic axis and the correctly-specified/misspecified axis:
\begin{itemize}
    \item \textbf{Well Specified:} Each method receives the exact per-sample
    $\phi_i$ used to generate the data.

    \item \textbf{Homoscedastic Misspecified:} Each method receives the constant $\widehat{\phi} = \bar{\phi} = 0.106$, equal to the true
    population mean. This condition tests whether discarding
    within-sample variance information is harmful.
    
    \item \textbf{Heteroscedastic Misspecified:} Each method receives noisy
    per-sample estimates
    \begin{equation}
        \widehat{\phi}_i = \phi_i \exp\!\bigl(e_i -
        \tfrac{\delta^2}{2}\bigr), \quad
        e_i \sim \mathcal{N}(0,\delta^2),
    \end{equation}
    where $\delta = \alpha\sqrt{\pi/2}$ is chosen so that the expected
    relative error satisfies
    $\mathbb{E}[|\widehat{\phi}_i/\phi_i - 1|] \approx \alpha = 0.25$.
    The log-normal centering ensures $\mathbb{E}[\widehat{\phi}_i] =
    \phi_i$ (mean-unbiased estimates).
\end{itemize}

The misspecified estimates $\{\widehat{\phi}_i\}$ are generated once per trial in the experiment from a deterministic fold of the data key to ensure all four methods see identical estimates within a trial.
\paragraph{Evaluation Metric}

All methods are evaluated with the \emph{sliced 1-Wasserstein distance} (SWD) between 10\,000 samples drawn from the fitted model and the 10\,000 held-out clean latent samples.
approximated with $P = 1000$ uniformly random projection directions,
using the \texttt{ott-jax} implementation with Euclidean ground cost.
Exact $W_2$ distances are avoided because their sample complexity scales
as $\mathcal{O}(n^{-1/D})$, making them impractical for $D \geq 7$.

\paragraph{Method Configurations}

\subparagraph{convMMD.}
The normalizing flow uses 6 coupling blocks, each consisting of a
Rational Quadratic Spline (RQS) coupling layer 
interleaved with an invertible linear layer (LU decomposition).  Each
coupling network has hidden size 1024 and 16 spline bins; the
tail bound is 1.0.  The base distribution mean and standard deviation
are initialised to the empirical mean and std of the training batch.
Training uses AdamW with cosine learning-rate decay from $5\times10^{-4}$
to $5\times10^{-6}$, gradient clipping at norm 1.0, weight decay $10^{-4}$,
and a full-batch size of 2\,000 for 3\,000 epochs.
The convMMD kernel bandwidth is set via a data-driven heuristic: nine
quantiles ($0.1, \ldots, 0.9$) of pairwise distances between training
observations and dummy model samples are computed, and the corresponding Laplace kernel widths are used in the convMMD loss.
\subparagraph{XDGMM \cite{bovy2011extreme}.}
We use a JAX implementation of Extreme Deconvolution
with $K = 25$ Gaussian components, EM run for up to 100 iterations
(tolerance $10^{-5}$), and a covariance regularisation of $10^{-6}\mathbf{I}$.
Components are initialised with a 10-iteration scikit-learn GMM fit to the observed (noisy) training data.  

\subparagraph{NPEB}
We use the NPEB package \citep{soloff2025multivariate} to fit the model. The package uses the Kiefer--Wolfowitz nonparametric maximum likelihood
estimator for Gaussian location mixtures, with
support atoms initialised by subsampling 2000 points from the training
data. 
\paragraph{deconv}
We use the stochastic variational inference normalizing flow from
\citep{dockhorn2020density}
with 5 flow steps, a batch size of 2000, Adam with learning rate
$5\times10^{-4}$, and 200 training epochs.  Noise covariances are passed as per-sample full diagonal matrices.

\subsection{Denoising MNIST Images}
\paragraph{Dataset: MNIST Images}
We utilize the MNIST dataset in this experiment. In each experimental run, models are trained utilizing the unlabelled noisy training data consisting of $N_{\text{train}} = 60\,000$ images. Quantitative evaluations are performed on a held-out test set consisting of $N_{\text{test}} = 1\,000$ images over 5 runs.  For $\text{convMMD}\_{GAN}$, images are maintained at their native $28 \times 28$ resolution ($D = 784$) and flattened to $D$-dimensional vectors. In $\text{convMMD}\_{GAN*}$, images are zero-padded to $32 \times 32$ ($D = 1024$). All image pixels are normalized to the range $[0, 1]$. 

\paragraph{Noise model.}
Each clean observation $\mathbf{x}_i$ is corrupted by an additive measurement error $\mathbf{u}_i$:
\begin{equation}
    \widetilde{\mathbf{x}}_i = \mathbf{x}_i + \mathbf{u}_i.
\end{equation}
We evaluate the models under two distinct measurement error regimes:
\begin{itemize}
    \item \textbf{Additive White Gaussian Noise (AWGN):} Independent pixel-wise errors drawn as $\mathbf{u}_i \sim \mathcal{N}(\mathbf{0}, \sigma^2 \mathbf{I})$ with a base standard deviation of $\sigma = 0.5$.
  
\item \textbf{Spatially Correlated AR(1) Noise:} Here, we generate stationary Autoregressive (AR(1)) Gaussian noise using exact spectral sampling. A white noise tensor is transformed to the frequency domain via 2D Fast Fourier Transform ($\mathcal{F}$), modulated by the square root of the power spectral density (PSD), and projected back to the spatial domain: 
\begin{equation}
    \mathbf{u}_i = \mathcal{F}^{-1} \Bigl( \sqrt{\mathbf{S}} \odot \mathcal{F}(\mathbf{W}) \Bigr),
\end{equation}
where $\mathbf{W} \sim \mathcal{N}(\mathbf{0}, \mathbf{I})$ is a white noise tensor, and $\odot$ denotes element-wise multiplication. The matrix $\mathbf{S}$ represents the discrete PSD evaluated at spatial frequencies $(\omega_x, \omega_y)$, given by:
\begin{equation}
    S(\omega_x, \omega_y) \propto \frac{1}{(1 - 2\rho \cos \omega_x + \rho^2)(1 - 2\rho \cos \omega_y + \rho^2)}.
\end{equation}
Here, $\rho \in [0, 1)$ is the spatial correlation parameter. To guarantee a strictly controlled comparison across varying correlation levels $\rho$, we apply exact per-sample normalization. Each synthesized noise sample $\mathbf{u}_i$ is scaled such that its sample standard deviation exactly equals the target $\sigma$ prior to addition with the clean image.
\end{itemize}

\paragraph{Method Configurations}

\subparagraph{convMMD-GAN (Fixed Kernel).}
This model parameterizes the prior using a Deep Convolutional Generator (DCGAN) mapping a latent vector $\mathbf{z} \sim \mathcal{N}(\mathbf{0}, \mathbf{I}_{10})$ to $\mathbb{R}^{784}$. The network comprises a dense layer mapping to a $256 \times 7 \times 7$ tensor, followed by a sequence of \texttt{ConvTranspose2d} blocks with SELU/ReLU activations and Batch Normalization, ending with a \texttt{Sigmoid} output. The model is trained using the convMMD objective matching the simulated noise-convolved samples against the observed noisy images. For this experiment, we utilize a mixture of Gaussian (RBF) kernels with bandwidths selected using median heuristic used in Section \ref{sec:2d_denoising_details}. Optimization is performed using Adam (learning rate $5\times 10^{-4}$, $\beta_1=0.5, \beta_2=0.999$) with gradient clipping at norm $10.0$ and a MultiStepLR scheduler (halving at iterations 15\,000 and 30\,000) for 300 epochs.

\subparagraph{convMMD-GAN* (Adversarial Kernel).}
This variant uses MMD-GAN minimax framework \citep{li2017mmd} for optimization.  The generator maps $\mathbf{z} \in \mathbb{R}^{10}$ to a $32 \times 32$ image. A critic network, structured as a multi-layer Convolutional Autoencoder, extracts feature representations of the data. The convMMD objective is evaluated in the critic's feature space using RBF bandwidths $\sigma_k \in \{1, 2, 4, 8, 16\}$. The network is trained using a Rank Hinge Loss combined with an Autoencoder $L_2$ reconstruction penalty (weights $\lambda_{\text{rg}}=16.0, \lambda_{\text{AE}}=8.0$) to prevent representation collapse. Both generator and critic are optimized using RMSprop (learning rate $5\times 10^{-5}$) for 250 epochs, with 5 critic updates per generator update.

\subparagraph{Noise2Self.}
We implement the self-supervised blind-spot network utilizing a fully convolutional U-Net architecture \citep{batson2019noise2self}. We enforce the blind-spot property using $J$-invariant masking strategy. During each training iteration, the input tensor is interpolated using a mask width of 4. The network is trained using Mean Squared Error (MSE) computed strictly over the masked pixels, preventing the learning of an identity mapping. The model is optimized using Adam (learning rate $10^{-3}$) for 50 epochs.

\subparagraph{SURE.}
Stein's Unbiased Risk Estimate (SURE) \citep{soltanayev2018training} provides an unsupervised loss function that requires no clean targets and a known noise variance. We parameterize the denoising model using a fully convolutional autoencoder. The encoder consists of two \texttt{Conv2d} layers (kernel size 3, stride 2, 10 channels), and the decoder symmetrically utilizes two \texttt{ConvTranspose2d} layers. All weights are initialized via Xavier Uniform, and all activations are \texttt{Sigmoid}. The SURE loss divergence term is approximated via Monte Carlo SURE (MC-SURE) using Gaussian perturbation vectors with $\epsilon = 10^{-3}$. The network is trained via Adam (learning rate $10^{-3}$) for 30 epochs.

\subparagraph{BUIFD.}
The Blind Unsupervised Image Filtering (BUIFD) architecture \citep{el2020blind} serves as our empirical upper bound. It consists of three highly parameterized parallel modules: (1) A 20-layer \texttt{DnCNN} feature extractor (64 channels per layer with \texttt{Conv2d}, \texttt{BatchNorm}, and \texttt{ReLU}) to estimate the additive noise. (2) A 6-layer \texttt{NoiseCNN} terminating in a \texttt{Sigmoid} activation to spatially estimate the standard deviation noise map. (3) A \texttt{FinalFusion} module consisting of three dilated convolutional layers that concatenates the noisy input, the estimated prior (input minus estimated noise), the estimated noise map, and their interactions to produce the final denoised image. The model is trained in a fully supervised configuration using MSE against exact ground-truth noise residuals for 50 epochs via Adam (learning rate $10^{-3}$).

\paragraph{Empirical Bayes Inference (convMMD variants)}
For both convMMD models, denoising is separated entirely from prior estimation. Once the generative models are trained, denoising is executed via self-normalized Importance Sampling to evaluate the posterior mean $\mathbb{E}[\mathbf{x} \mid \widetilde{\mathbf{x}}]$. For each noisy test image, we draw $S = 100\,000$ latent samples $\mathbf{z}_s$ and generate their corresponding clean candidate images $\mathbf{x}_s$. Candidate weights are assigned proportional to the exact known noise likelihood $m(\widetilde{\mathbf{x}} - \mathbf{x}_s)$. 

For the spatially correlated AR(1) noise condition, direct spatial evaluation of the dense $D \times D$ covariance likelihood is computationally intractable. We resolve this by computing the log-likelihood in the frequency domain. Utilizing Parseval's theorem, we take the 2D FFT of the spatial residual $\widetilde{\mathbf{x}} - \mathbf{x}_s$ and weight the squared magnitude by the exact inverse AR(1) PSD.

\newpage

\end{document}

%% file: sections/introduction.tex
\section{Introduction}

Modern scientific inference increasingly relies on measurements that are only corrupted proxies for the latent quantities of interest, so ignoring measurement error can cause an analysis to target the observed surrogate rather than the scientific object itself \citep{carroll2006measurement}. In survey astronomy, for example, stellar mass, luminosity, and weak-lensing shear must be inferred from observations degraded by instrumental noise, selection effects, projection distortions, and model-dependent pipelines \citep{starck2002deconvolution, sok2022finite, anbajagane2025decade}. These distortions can change scientific conclusions: small biases in recovered weak-lensing shape distributions can propagate into biased cosmological inference \citep{mandelbaum2018weak}, while noisy stellar-mass estimates can induce Eddington bias, scattering more moderate-mass galaxies into high-mass bins than the reverse and altering the inferred abundance and high-mass slope of the galaxy stellar mass function \citep{obreschkow2018eddington}. More broadly, measurement error is pervasive and has contributed to reproducibility failures across scientific domains \citep{loken2017measurement}. Thus, deconvolution is not merely a technical correction but a prerequisite for aligning statistical learning with the latent distributions that science seeks to understand. 

This population-level perspective naturally leads to a closely related and equally important problem: \textit{denoising}. When the scientific target is to unravel information about individual objects, such as a particular image, spectrum, galaxy, experimental sample, or data point, it is not enough to recover the latent distribution alone \citep[e.g.,][]{richardson1972bayesian, lucy1974iterative, pu2016deep, puetter2005digital, dockhorn2020density, sedaghat2021machines, wang2019data, zhang2023hyperspectral}. One must also infer the latent realization underlying each corrupted observation. Denoising, therefore, complements deconvolution by moving from distributional recovery to object-level recovery, enabling scientific interpretation and downstream inference at the level of individual measurements.

In this work, we study this problem in a multivariate, potentially high-dimensional regime under the additive noise model $\widetilde{\mathbf{X}} = \mathbf{X} + \mathbf{U}$,
where $\mathbf{X}$ is a latent variable, $\mathbf{U}$ is measurement noise, and only noisy observations $\widetilde{\mathbf{X}}$ are available \citep{carroll1988optimal,liu1989consistent,stefanski1990deconvolving, fan1991optimal, delaigle2008density}. 
This setting induces two coupled tasks: \emph{density deconvolution}, which seeks to recover the distribution of $\mathbf{X}$, and \emph{denoising}, which seeks to estimate latent realizations $\mathbf{X}_{i}$ from individual observations $\widetilde{\mathbf{X}}_{i}$. Both are ill-posed, and their difficulty depends strongly on the noise distribution, with high-frequency structure often the hardest to recover \citep{carroll1988optimal, fan1991optimal}.

Classical deconvolution methods often rely on Fourier inversion, kernel estimators, orthogonal and trigonometric series expansions \citep{carroll1988optimal, liu1989consistent, stefanski1990deconvolving, fan1991optimal, delaigle2008density, hall2005discrete}. % 
While these approaches offer a mathematically rigorous foundation, illuminating phenomena such as the phase transition between ordinary and super-smooth distributions, they struggle in practice. 
Specifically, they often require meticulous regularization, scale poorly with increasing dimensions, and fail to leverage the rich, learned representations available in data. 
Parametric models, though tractable, remain susceptible to performance degradation under misspecification. 
Consequently, bridging the gap between classical deconvolution theory and high-capacity generative modeling remains an open challenge.

Recent progress in deep generative models offers a route to close this gap. Normalizing flows, variational models, diffusion-based methods, and related approaches can represent highly structured distributions in spaces where classical density estimators struggle \citep{rezende2015variational, kingma2018glow, papamakarios2021normalizing, lu2025stochastic}. Several recent methods bring such models to  deconvolution: \citep{dockhorn2020density} use stochastic variational inference with normalizing flows and \citep{lu2025stochastic} trains diffusion models via stochastic forward--backward iteration on noisy data. Yet most training objectives assume clean samples or absorb corruption into the learned distribution, and likelihood-based training under known noise can become intractable when the noisy-data likelihood requires high-dimensional integration or unstable inversion. The central question is therefore how to train expressive generative models directly on noisy observations while preserving statistical interpretability and avoiding explicit deconvolution. Closely related multivariate baselines include NPMLE \citep{soloff2025multivariate} and Extreme Deconvolution \citep{bovy2011extreme}. 

We address this question through a simulation-based, likelihood-free framework based on \emph{Convolutional Maximum Mean Discrepancy} \citep[convMMD,][]{convMMD}. When the corruption mechanism is known, we compare the empirical distribution of the noisy data with that obtained by convolving a candidate latent model with the same noise law. This avoids direct inversion and bypasses evaluation of a deconvolved likelihood, replacing unstable spectral inversion with convolved distribution matching in observation space.
Our method has two stages. We first learn a latent density model $q_{\boldsymbol{\theta}}$ by minimizing an empirical convMMD criterion between the noisy observations and the noise-convolved model. This objective requires only sampling from the latent model and simulator of the corruption mechanism, making it compatible with stochastic gradient optimization and flexible generator families. We then use the fitted latent model as an empirical prior and perform posterior inference for each observation, yielding an \emph{Empirical Bayes} denoiser \cite{robbins1992empirical, efron2019bayes}. A single learned model therefore supports both latent density estimation and pointwise recovery. 

Our contributions are both methodological and theoretical. Methodologically, we cast nonparametric deconvolution and denoising as a unified generative learning problem, minimizing the convMMD objective over flexible sieve classes like Gaussian mixtures and neural generative models. Theoretically, we extend prior parametric work \citep{convMMD} to the nonparametric regime. We establish finite-sample oracle inequalities, proving that the estimator achieves a fast parametric $\mathcal{O}(N^{-1/2})$ convergence rate in the latent reproducing kernel Hilbert space (RKHS) geometry. Furthermore, we prove nonparametric $L_2$ convergence rates for the latent density across both ordinary- and super-smooth noise regimes. Finally, we connect these results to downstream inference by linking latent density estimation error to posterior recovery risk.

%% file: sections/theory.tex
\vspace{-5mm}
\section{Method} \label{sec:method}
\vspace{-3mm}
We propose a unified framework for nonparametric density deconvolution and empirical Bayes denoising under additive measurement error. The method fits a flexible generative model for the latent density by matching noisy observations to the model-implied noise-convolved distribution, then uses the fitted density as an empirical prior for denoising. We formalize the inverse problem, introduce a simulation-based convMMD estimator over growing generative sieves, prove finite-sample and asymptotic guarantees, and demonstrate the approach with Gaussian mixtures and generative models. 

\subsection{Problem Setup: Deconvolution and Denoising}
\label{subsec:setup}

%We propose a unified framework for nonparametric density deconvolution and empirical Bayes denoising under additive measurement error. The method estimates the latent density by matching the empirical distribution of the noisy observations to the corresponding noise-convolved distribution induced by a flexible generative model, then uses the fitted latent density as an empirical prior for denoising. We formalize the inverse problem, define a simulation-based convMMD estimator over a growing sieve of generative models, establish finite-sample and asymptotic guarantees, and illustrate the framework with Gaussian mixture models and normalizing flows. 

%\subsection{Problem Setup: Deconvolution and Denoising}

Let $\mathbf{X} \in \mathcal{X} \subseteq \mathbb{R}^d$ be an unobserved latent random variable with unknown density $p$. We observe only a noisy proxy
$$
    \widetilde{\mathbf{X}} = \mathbf{X} + \mathbf{U}_X,
$$
where $\mathbf{U}_X$ is an additive measurement error with a known distribution $m$. We denote the distribution of the noisy observation by $\widetilde{p}$, so that
$
    \widetilde{p} = p * m.
$
Given i.i.d.\ noisy observations
$
    \mathcal{D}_N = \{\widetilde{\mathbf{x}}_i\}_{i=1}^N \stackrel{\mathrm{iid}}{\sim} \widetilde{p},
$
our objective is two-fold. The first goal is \emph{deconvolution}: estimate the unknown latent density $p$ from the noisy sample $\mathcal{D}_N$. We do so by introducing a model family $\{q_{\boldsymbol{\theta}} : \boldsymbol{\theta} \in \Theta_J\}$ and learning $q_{\boldsymbol{\theta}}$ so that its noise-convolved distribution $\widetilde{q}_{\boldsymbol{\theta}} = q_{\boldsymbol{\theta}} * m$ matches the observed noisy law. The second goal is \emph{denoising}: once a latent density estimate has been learned, recover the latent value associated with a particular noisy observation $\widetilde{\mathbf{x}}_i$. Under the learned model $q_{\widehat{\boldsymbol{\theta}}}$, the posterior density of the latent signal is
$$
    \pi_{\widehat{\boldsymbol{\theta}}}(\mathbf{x} \mid \widetilde{\mathbf{x}}_i)
    \propto
    m(\widetilde{\mathbf{x}}_i - \mathbf{x})\, q_{\widehat{\boldsymbol{\theta}}}(\mathbf{x}),
$$
and the denoised estimate is taken to be the posterior mean
$$
    \widehat{\mathbf{x}}_i
    =
    \mathbb{E}_{q_{\widehat{\boldsymbol{\theta}}}}[\mathbf{X} \mid \widetilde{\mathbf{X}}=\widetilde{\mathbf{x}}_i].
$$
Thus, density estimation and denoising are treated within a single framework: first estimate the latent population distribution, then use that estimate to regularize pointwise inversion.

\subsection{convMMD Estimation over Sieve Classes}
\label{subsec:convmmd}

Let $k$ be a bounded positive definite kernel on $\mathbb{R}^d$. For a candidate latent density $q_{\boldsymbol{\theta}}$, we define its induced noisy distribution by $q_{\boldsymbol{\theta}} * m$. We estimate $\boldsymbol{\theta}$ by minimizing the MMD \citep[][]{gretton12a} between the empirical noisy distribution and the model-implied noisy distribution:
\begin{equation}
\label{eq:estimator}
\widehat{\boldsymbol{\theta}}_N
=
\arg\min_{\boldsymbol{\theta} \in \Theta_J}
\widehat{\operatorname{convMMD}}_k^2(p,q_{\boldsymbol{\theta}},m)
\equiv
\arg\min_{\boldsymbol{\theta} \in \Theta_J}
\operatorname{MMD}_k^2\!\left(\widehat{(p*m)}_N, q_{\boldsymbol{\theta}} * m\right).
\end{equation}
With $\widetilde{\mathbf{Y}} \sim q_{\boldsymbol{\theta}} * m$, the corresponding population objective takes the form
{\footnotesize
\begin{align}
    L_N(\boldsymbol{\theta})
    &\equiv
    \operatorname{MMD}_k^2\!\left(\widehat{(p*m)}_N, q_{\boldsymbol{\theta}} * m\right)
    =
    \frac{1}{N^2}\sum_{i,j=1}^N k(\widetilde{\mathbf{x}}_i,\widetilde{\mathbf{x}}_j)
    +
    \mathbb{E}_{\widetilde{\mathbf{Y}},\widetilde{\mathbf{Y}}'}[k(\widetilde{\mathbf{Y}},\widetilde{\mathbf{Y}}')]
    -
    \frac{2}{N}\sum_{i=1}^N \mathbb{E}_{\widetilde{\mathbf{Y}}}[k(\widetilde{\mathbf{x}}_i,\widetilde{\mathbf{Y}})].
\label{eq:objective}
\end{align}}
The objective is likelihood-free and simulation-based. If one can sample $\mathbf{Y} \sim q_{\boldsymbol{\theta}}$ and independently simulate $\mathbf{U}_Y \sim m$, then $\widetilde{\mathbf{Y}}=\mathbf{Y}+\mathbf{U}_Y$ is a draw from $\widetilde{q}_{\boldsymbol{\theta}} = q_{\boldsymbol{\theta}} * m$. So the expectations in Equation~\eqref{eq:objective} can be approximated by Monte Carlo, and an optimization algorithm can then be used to arrive at $\widehat{\boldsymbol{\theta}}_N$. See Appendix~\ref{app:algorithm} for the algorithm.

To move beyond the parametric setting, we place the estimator in a sieve framework. Let
$$
    \mathcal{Q}_J = \{q_{\boldsymbol{\theta}} : \boldsymbol{\theta} \in \Theta_J\},
    \qquad
    \mathcal{Q}_1 \subset \mathcal{Q}_2 \subset \cdots,
$$
where $J$ is a complexity index, and assume that $\cup_{J \ge 1} \mathcal{Q}_J$ is dense in a target class of latent densities. The parameter $J$ can represent, for example, the number of mixture components or the width or depth of a neural architecture. This formulation allows the estimation problem to be analyzed separately at two levels: a general level covering any regular enough simulation-based sieve, and a sharper level requiring quantitative approximation bounds for the sieve in a strong norm such as $L_2$.

\subsection{Empirical Bayes Denoising}
\label{subsec:eb}

Once the latent density estimate $q_{\widehat{\boldsymbol{\theta}}_{N}}$ is obtained, we use it as an empirical prior for denoising. For each noisy observation $\widetilde{\mathbf{x}}_{i}$, the posterior distribution is
\begin{equation}
    \pi_{\widehat{\boldsymbol{\theta}}_N}(\mathbf{x} \mid \widetilde{\mathbf{x}}_i)
    =
    \frac{m(\widetilde{\mathbf{x}}_i-\mathbf{x}) q_{\widehat{\boldsymbol{\theta}}_N}(\mathbf{x})}
    {\int m(\widetilde{\mathbf{x}}_i-\mathbf{z}) q_{\widehat{\boldsymbol{\theta}}_N}(\mathbf{z})\,d\mathbf{z}},
\end{equation}
and the empirical Bayes estimator is the posterior mean
\begin{equation}
    \widehat{\mathbf{x}}_i^{\mathrm{EB}}
    =
    \int \mathbf{x}\, \pi_{\widehat{\boldsymbol{\theta}}_N}(\mathbf{x} \mid \widetilde{\mathbf{x}}_i)\,d\mathbf{x}.
\end{equation}
In practice, if the density $q_{\widehat{\boldsymbol{\theta}}_N}$ can be evaluated, this posterior mean can be computed by numerical integration or MCMC. If the generative model $q_{\widehat{\boldsymbol{\theta}}_N}$ is implicit and does not admit a tractable density (e.g., a generative adversarial network), the posterior mean can still be efficiently approximated via self-normalized importance sampling in the latent space, drawing samples $\mathbf{Z}$ from the base distribution and weighting the corresponding generated outputs by the known noise likelihood.  
\section{Theoretical Guarantees}
\label{subsec:theory-general}
We now state the basic guarantees that hold for any sieve class satisfying the assumptions collected in Appendix~\ref{app:assumptions}. The first theorem gives a finite-sample oracle inequality in convMMD loss. The second states the latent-space equivalence established in \citep{convMMD}. As a result, this method yields a parametric-rate guarantee in the latent RKHS geometry induced by the noise-smoothed kernel.

\begin{theorem}[Finite-sample oracle inequality for sieve-convMMD]
\label{th:oracle}
Under Assumptions~\ref{assump:noise-indep}, \ref{assump:known-noise}, \ref{assump:noise-second-moment}, \ref{assump:kernel}, and \ref{assump:compactness}, 
for every $\delta \in (0,1)$, with probability at least $1-\delta$,
\begin{equation}
\label{eq:oracle-bound}
\operatorname{convMMD}_k\!\left(p,q_{\widehat{\boldsymbol{\theta}}_N},m\right)
\le
\inf_{\boldsymbol{\theta} \in \Theta_J}
\operatorname{convMMD}_k(p,q_{\boldsymbol{\theta}},m)
+
2\sqrt{\frac{K}{N}}
+
2\sqrt{\frac{2K\log(1/\delta)}{N}},
\end{equation}
where $K$ is the kernel bound from Assumption~\ref{assump:kernel}. 
\end{theorem}

\begin{remark}
Theorem~\ref{th:oracle} shows that the learned model performs, in convMMD loss, almost as well as the best element of the sieve class. This theorem is model-agnostic. It does not require mixtures, flows, or any particular parametric form. It applies to any model family for which $q_{\boldsymbol{\theta}}$ is a valid density, the map $\boldsymbol{\theta} \mapsto q_{\boldsymbol{\theta}} * m$ is well-defined, and optimization is restricted to a regularized compact sieve.
\end{remark}

\begin{theorem}[Latent MMD equivalence, \citealp{convMMD}]
\label{th:equiv}
Suppose Assumptions~\ref{assump:noise-indep}, \ref{assump:known-noise}, \ref{assump:noise-second-moment}, \ref{assump:conv-invertibility}, and \ref{assump:kernel} hold. Let $k(\mathbf{x},\mathbf{y})=k(\mathbf{x}-\mathbf{y})$ be translation invariant. Define the smoothed kernel
$$
    \widetilde{k}(\mathbf{x},\mathbf{y})
    =
    \mathbb{E}_{\mathbf{U}_X,\mathbf{U}_Y \sim m}[k(\mathbf{x}+\mathbf{U}_X,\mathbf{y}+\mathbf{U}_Y)].
$$
Then, for any two latent distributions $p$ and $q$,
\begin{equation}
    \operatorname{convMMD}_k(p,q,m)
    =
    \operatorname{MMD}_{\widetilde{k}}(p,q).
\end{equation}
\end{theorem}

\begin{corollary}[Parametric convergence in latent RKHS]
\label{cor:latent}
Under the assumptions of Theorems~\ref{th:oracle} and \ref{th:equiv}, for every $\delta \in (0,1)$, with probability at least $1-\delta$,
\begin{equation}
    \operatorname{MMD}_{\widetilde{k}}(p,q_{\widehat{\boldsymbol{\theta}}_N})
    \le
    \inf_{\boldsymbol{\theta} \in \Theta_J}
    \operatorname{MMD}_{\widetilde{k}}(p,q_{\boldsymbol{\theta}})
    +
    2\sqrt{\frac{K}{N}}
    +
    2\sqrt{\frac{2K\log(1/\delta)}{N}}.
\end{equation}
In particular, if the sieve is well-specified at level $J$, then
\begin{equation}
    \operatorname{MMD}_{\widetilde{k}}(p,q_{\widehat{\boldsymbol{\theta}}_N})
    =
    O_p(N^{-1/2}).
\end{equation}
\end{corollary}

\begin{remark}
Corollary~\ref{cor:latent} gives a sharp parametric rate, but in the latent RKHS geometry associated with the smoothed kernel $\widetilde{k}$ and with a sieve with capacity $J$ large enough to model the distribution $p$. This is weaker than an $L_2$ statement and does not by itself quantify full recovery of high-frequency structure lost through convolution \citep{sriperumbudur2010hilbert, liu2017approximation}.  To obtain such guarantees, one must impose stronger smoothness and approximation assumptions. 
\end{remark}

\subsection{Nonparametric Guarantees}
\label{subsec:theory-strong}

We now state the stronger theorem that yields explicit nonparametric convergence rates in $L_2$. This result requires additional assumptions on the smoothness of the latent density, the ill-posedness of the noise, and the approximation power of the sieve. These assumptions are collected in Appendix~\ref{app:assumptions} and are substantially stronger than those used for the basic convMMD guarantees.

\begin{theorem}[Nonparametric $L_2$ Convergence] \label{th:l2rate}
  Under Assumptions~\ref{assump:noise-indep} --  \ref{assump:compactness}, let $q_{\widehat{\boldsymbol{\theta}}_N}$ be the convMMD estimator over $\mathcal{Q}_J$. 

  \textbf{Case A (Ordinary Smooth Noise and Kernel):} Under (OS-N) with exponent $\gamma$ and (OS-K) with exponent $\nu$, choosing $J_N \asymp N^{d/(2\beta)}$ yields
  $$\|q_{\widehat{\boldsymbol{\theta}}_N} - p\|_{2} = O_p\left(N^{-\beta/(2\beta + 2\gamma + \nu)}\right).$$

  \textbf{Case B (Supersmooth Noise or Kernel):} If the noise is supersmooth (SS-N) with exponent $\gamma$, or the kernel is supersmooth (SS-K) with exponent
  $\nu$, or both, define
  $$\gamma^* = \begin{cases}
  \gamma & \text{if (SS-N) and (OS-K)}, \\
  \nu & \text{if (OS-N) and (SS-K)}, \\
  \max(\gamma, \nu) & \text{if (SS-N) and (SS-K)}.
  \end{cases}$$
  Then, choosing $J_N \asymp N^{d/(2\beta)}$ yields
  $$\|q_{\widehat{\boldsymbol{\theta}}_N} - p\|_{2} = O_p\left((\log N)^{-\beta/\gamma^*}\right).$$
\end{theorem}

% Theorem~\ref{th:l2rate} makes clear that the stronger $L_2$ theory applies only to sieve classes satisfying a quantitative approximation bound. Universality property of a sieve classes alone is not sufficient. The theorem also reveals a feature specific to convMMD-based deconvolution: the rate depends on the smoothness of the latent density, the noise, and the Fourier decay of the kernel used in the MMD loss. In the ordinary-smooth regime, choosing a kernel with lighter Fourier smoothing improves the rate. In the super-smooth regime, the exponential ill-posedness of the noise dominates. We also note that the two regimes in Theorem~\ref{th:l2rate} correspond to the classical distinction in deconvolution theory between polynomially decaying and exponentially decaying characteristic functions \red{[citation?]}.

The two regimes in Theorem~\ref{th:l2rate} correspond to the classical distinction between \emph{ordinary-smooth} and \emph{super-smooth} inverse problems, determined by how quickly the noise attenuates high-frequency information. In the ordinary-smooth regime, the noise characteristic function decays polynomially, leading to a mildly ill-posed problem and polynomial $L_2$ convergence rates. In contrast, in the super-smooth regime, exponential decay severely suppresses high frequencies, making the problem exponentially ill-posed and degrading convergence to logarithmic rates. 

A key feature of convMMD is that the ordinary-smooth rate depends on the kernel's Fourier decay $\nu$. Compared to the classical minimax optimal bounds \citep{fan1991optimal}, our exponent replaces the dimension $d$ with $\nu$. Since bounded MMD kernels require $\nu > d$, our rate is slightly sub-optimal. However, this theoretical gap reflects a deliberate practical trade-off: achieving optimal rates  \citep{fan1991optimal} via classical density estimation requires vanishing bandwidths that cause numerical instability. By using a fixed, smooth kernel, convMMD incurs a small asymptotic speed penalty to guarantee a stable optimization landscape for generative models. In the super-smooth regime, the noise dominates, and standard logarithmic rates are recovered. We defer detailed examples to Appendix~\ref{app:smooth-examples}.

\subsection{Empirical Bayes Denoising Guarantees}
\label{subsec:denoising}

We connect density estimation to empirical Bayes denoising. Let $\widehat{\mathbf{x}}^*$ be the oracle posterior mean (using the true $p$) and $\widehat{\mathbf{x}}^{\mathrm{EB}}$ be the empirical Bayes estimator (using $q_{\widehat{\boldsymbol{\theta}}_N}$). 

% \begin{theorem}[Empirical Bayes Denoising Rates]
% \label{th:denoising-general-rates}
% Under Assumptions~\ref{assump:noise-indep} --  \ref{assump:self-reg}, let $\widehat{\mathbf{x}}^*$ be the oracle posterior mean and $\widehat{\mathbf{x}}^{\mathrm{EB}}$ be the empirical Bayes estimator using $q_{\widehat{\boldsymbol{\theta}}_N}$. Combining this with Theorem~\ref{th:l2rate}, the expected excess risk scales as $O_p(N^{-2\beta/(2\beta + 2\gamma + \nu)})$ for ordinary-smooth noise, and $O_p((\log N)^{-2\beta/\gamma^*})$ for super-smooth noise.
% \end{theorem}

\begin{theorem}[Empirical Bayes Denoising Rates]
\label{th:denoising-general}
Let $\widehat{\mathbf{x}}^*$ be the oracle posterior mean and $\widehat{\mathbf{x}}^{\mathrm{EB}}$ be the empirical Bayes estimator using $q_{\widehat{\boldsymbol{\theta}}_N}$. Under Assumptions~\ref{assump:noise-indep} -- \ref{assump:bounded-noise},  with the additional requirement that the true marginal $\widetilde{p}$ is strictly positive on a compact observation domain $\mathcal{K}$, the expected excess empirical Bayes risk evaluated over $\mathcal{K}$ satisfies:
$$ \int_{\mathcal{K}} \|\widehat{\mathbf{x}}^{\mathrm{EB}}(\widetilde{\mathbf{x}}) - \widehat{\mathbf{x}}^*(\widetilde{\mathbf{x}})\|^2 \widetilde{p}(\widetilde{\mathbf{x}}) d\widetilde{\mathbf{x}} = O_p\!\left(\|q_{\widehat{\boldsymbol{\theta}}_N} - p\|_{2}^2\right). $$
Combined with Theorem~\ref{th:l2rate}:
\begin{itemize}
    \item \textbf{Case A (Ordinary-smooth):} $O_p\Big(N^{-\frac{2\beta}{2\beta + 2\gamma + \nu}}\Big)$.
    \item \textbf{Case B (Super-smooth):} $O_p\Big((\log N)^{-\frac{2\beta}{\gamma^*}}\Big)$.
\end{itemize}
\end{theorem}
Theorem~\ref{th:denoising-general} shows that accurate deconvolution implies accurate denoising: the excess Bayes risk inherits the squared $L_2$ rate of the latent density estimator. When the noise is Gaussian ($m = \mathcal{N}(\mathbf{0},\boldsymbol{\Sigma})$), Tweedie's formula \citep{robbins1992empirical, efron2011tweedie} reveals that $\widehat{\mathbf{x}}^*(\widetilde{\mathbf{x}}) = \widetilde{\mathbf{x}} + \boldsymbol{\Sigma} \nabla \widetilde{p}(\widetilde{\mathbf{x}})/\widetilde{p}(\widetilde{\mathbf{x}})$. Denoising then depends only on the accuracy of the marginal density, bypassing deconvolution entirely, and it may be possible to obtain faster rates of convergence in this case.

\subsection{Model Classes Covered by the Theory}
\label{subsec:model-classes}

The preceding theorems apply at two levels of generality. At the first level, Theorem~\ref{th:oracle}, Theorem~\ref{th:equiv}, and Corollary~\ref{cor:latent} apply to any \emph{simulation-based sieve model} satisfying the regularity assumptions in Appendix~\ref{app:assumptions}. The model class must define valid probability measures from which one can sample, admit simulation from the convolved model $q_{\boldsymbol{\theta}} * m$, and be regularized over a compact parameter set at each sieve level. No particular parametric form is required. At the second level, Theorem~\ref{th:l2rate} additionally requires quantitative $L_2$ sieve approximation rates (Assumption~\ref{assump:sieve-approx}). This is a mathematical property that separates model families for which the full deconvolution theory is currently available from those for which it remains conjectural.
We illustrate this distinction using two classes of models.

\paragraph{Gaussian Mixture Models (GMMs).} 
Gaussian mixture models provide a classical sieve for which the assumptions can be verified explicitly \citep{genovese2000rates, ghosal2001entropies}, including quantitative approximation rates. Consider the sieve
$
q_{\boldsymbol{\theta}}(\mathbf{x}) = \sum_{j=1}^J \pi_j \mathcal{N}(\mathbf{x} \mid \boldsymbol{\mu}_j, \boldsymbol{\Sigma}_j)
$,
with parameters $\boldsymbol{\theta} = \{\pi_j, \boldsymbol{\mu}_j, \boldsymbol{\Sigma}_j\}_{j=1}^J$, weights $\pi_j \ge 0$, and $\sum_{j=1}^J \pi_j = 1$. A sieve of $J$-component Gaussian mixtures with bounded means and regularized covariance matrices forms a compact, simulation-friendly parametric family. Under standard compactness and Sobolev regularity conditions on the target, classical results prove that GMMs achieve the exact polynomial $L_2$ approximation rates required by Assumption~\ref{assump:sieve-approx} \citep{genovese2000rates, ghosal2001entropies}. 

\paragraph{Neural Generative Models.} 
Models defined as neural pushforward maps $q_{\boldsymbol{\theta}} = (G_{\boldsymbol{\theta}})_\# p_z$, including normalizing flow models (NFM) \citep{papamakarios2021normalizing},  generative adversarial networks  (GANs)  \citep{goodfellow2014generative}, and diffusion models \cite{sohl2015deep, ho2020denoising, songscore}, are structurally ideal for convMMD. Simulating the noise-convolved distribution is mathematically trivial: $\widetilde{\mathbf{Y}} = G_{\boldsymbol{\theta}}(\mathbf{Z}) + \mathbf{U}_Y$, allowing even implicit models like GANs to be fitted without tractable densities. We note that the choice of architecture dictates the feasibility of downstream tasks. While implicit models like GANs are sufficient for empirical Bayes denoising, explicitly evaluating the learned density for pure deconvolution requires a tractable model such as normalizing flows. Under standard network regularization (e.g., spectral normalization), these models fully satisfy the basic convMMD estimation guarantees of Theorem~\ref{th:oracle}.  While these architectures are known universal approximators \citep{teshima2020coupling, hornik1989multilayer}, proving explicit $L_2$ polynomial approximation rates for their induced pushforward densities remains an active area of approximation theory. Consequently, Theorem~\ref{th:l2rate} applies to neural generative models conditionally, assuming the underlying architecture possesses the necessary approximation capacity.

A formal mathematical verification of how both GMMs and Neural Generative Models satisfy the framework's assumptions (including handling degenerate distributions and empirical Bayes importance sampling) is provided in Appendix~\ref{app:model_classes}.

%% file: sections/results.tex
\section{Experiments}\label{sec:exp}

\paragraph{Verifying Rates of Convergence.}

Figure~\ref{fig:rkhs-rate} provides empirical evidence for the two convergence phenomena predicted by our theory. We evaluate convMMD on a one-dimensional heteroscedastic deconvolution problem in which the latent density is a two-component Laplace mixture and the observations are corrupted by additive measurement noise. The estimator is trained directly by minimizing the convMMD objective, using two sieve classes: Gaussian mixture models and rational-quadratic-spline normalizing flows \cite{durkan2019neural}. We vary the sample size and compare two inverse-problem regimes: an ordinary-smooth setting with Laplace noise and a Laplace kernel, and a supersmooth setting with Gaussian noise and a Gaussian kernel.

\begin{figure}[h!]
    \centering
    \vspace{-3mm}
    \includegraphics[width=0.97\linewidth]{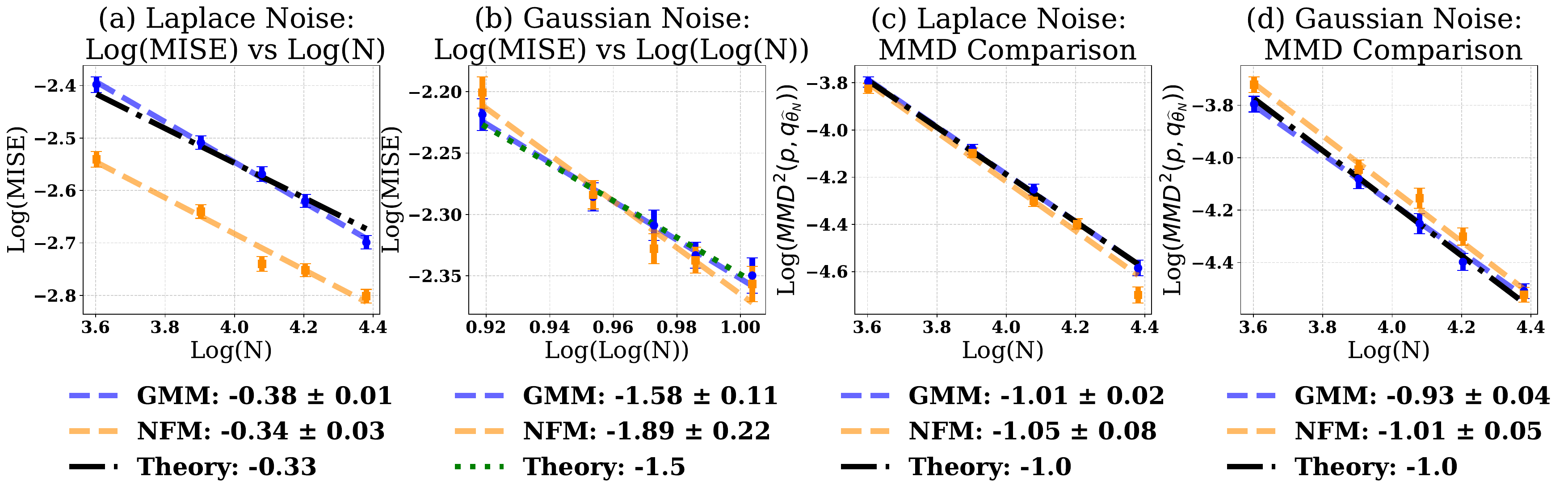}
    \caption{\textit{Rate validation.} $L_2$ error follows the ordinary-smooth and supersmooth rates predicted by Theorem~\ref{th:l2rate}, while Latent $\mathrm{MMD}^2$ follows the parametric $N^{-1}$ rate predicted by Corollary~\ref{cor:latent}}
    \label{fig:rkhs-rate}
    \vspace{-3mm}
\end{figure}

The right panels validate Corollary~\ref{cor:latent}. In both the ordinary-smooth and supersmooth settings, the latent squared MMD error, $\mathrm{MMD}^2(p,q_{\widehat{\theta}_N})$, decays at the parametric $N^{-1}$ rate, indicating that convMMD achieves fast convergence in the RKHS metric even under measurement error. The left panels validate Theorem~\ref{th:l2rate}. For ordinary-smooth noise, the $L_2$ density error follows the predicted polynomial decay; for supersmooth Gaussian noise, the rate slows to the expected logarithmic behavior. Thus, the experiments reproduce the central theoretical distinction: RKHS convergence remains parametric, while nonparametric density recovery reflects the ill-posedness of the deconvolution operator. Full experimental details are provided in Appendix~\ref{sec: exp-details}. \vspace{-2mm}

%\begin{figure}[H]
%    \centering
%    \includegraphics[width=0.98\linewidth]{Images/l2_rate.png}
%    \caption{Illustration of Theorem \ref{th:l2rate} for the special case of (a) Laplace Noise with Laplace Kernel and (b) Gaussian Noise with Gaussian Kernel.}
%    \label{fig:l2-rate}
%\end{figure}

\paragraph{Denoising in presence of outliers.}  We next demonstrate empirical Bayes denoising with our simulation-based framework. Given the learned sieve density $q_{\widehat{\theta}_N}$, we use it as a data-driven prior to estimate the posterior mean $\mathbb{E}[X_i \mid \tilde{X}_i]$. We evaluate this procedure on three structured two-dimensional datasets (Two Moons, Circle, and Checkerboard) where latent points are corrupted by heteroscedastic Gaussian noise. To assess robustness, we further contaminate the latent distribution with $3\%$ outliers. See Table~\ref{tab:results} for comparison of convMMD with normalizing flow (NF) and Gaussian mixture model (GMM) sieves against two standard deconvolution baselines: NPEB, a nonparametric empirical Bayes method based on the NPMLE \citep{soloff2025multivariate}, and XDGMM, which fits a Gaussian-mixture prior by expectation maximization~\citep{bovy2011extreme}.  An Oracle baseline is also included, which uses the true latent distribution estimated via KDE on a large clean sample, and thus provides an empirical upper bound on the achievable performance. 

\begin{figure}[H]
    \centering
     \vspace{-2mm}
    \includegraphics[width=0.97\linewidth]{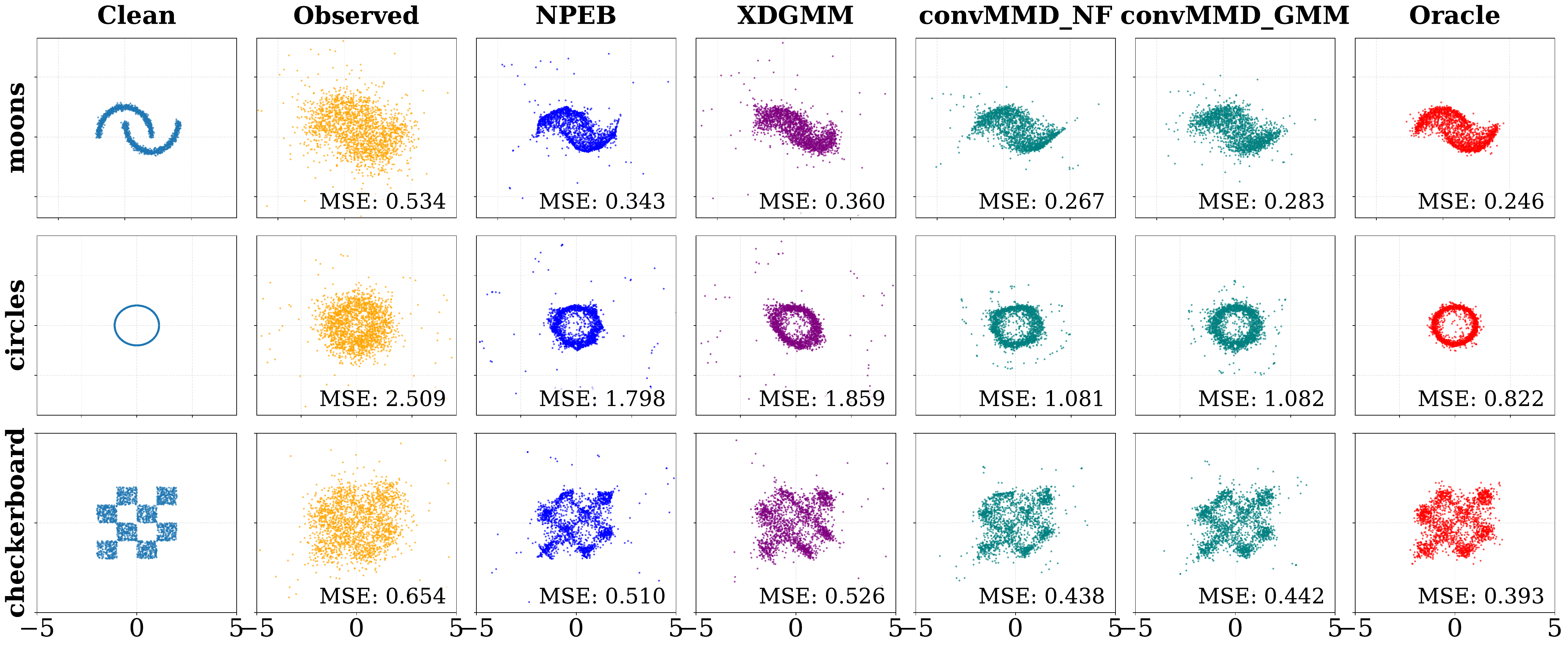}
    \caption{Benchmarking denoising. The last column assumes the true density is perfectly known. }
    \label{fig:2d-exp}
    \vspace{-2mm}
\end{figure}

\begin{wraptable}{r}{0.61\textwidth}
\vspace{-11pt}
\centering
\small

\caption{Performance on synthetic datasets. We report mean $\pm$ std of MSE and Wasserstein-1 (W1) for denoising and deconvolution, respectively.}
\label{tab:results}
\vspace{-2mm}
\resizebox{0.60\textwidth}{!}{
\begin{tabular}{llcc}
\toprule
\textbf{Dataset} & \textbf{Method} & \textbf{MSE ($\downarrow$)} & \textbf{W1 Distance ($\downarrow$)} \\ \midrule
\multirow{5}{*}{Moons} 
    & convMMD\_NF & $\mathbf{0.2672 \pm 0.0154}$ & $0.1242 \pm 0.0022$ \\
    & convMMD\_GMM & $0.2829 \pm 0.0071$ & $\mathbf{0.1148 \pm 0.0185}$ \\
    & NPEB    & $0.3428 \pm 0.0095$ & $0.2112 \pm 0.0050$ \\
    & XDGMM   & $0.3599 \pm 0.0089$ & $0.1896 \pm 0.0047$ \\
    & Oracle  & \textit{0.2460 $\pm$ 0.0040} & \textit{0.0350 $\pm$ 0.0019}\\ \midrule
\multirow{5}{*}{Circles} 
    & convMMD\_NF & $\mathbf{1.0809 \pm 0.0202}$ & $0.2935 \pm 0.0139$ \\
    & convMMD\_GMM & $1.0817 \pm 0.0238$ & $\mathbf{0.2717 \pm 0.0210}$ \\
    & NPEB    & $1.7982 \pm 0.0635$ & $0.5900 \pm 0.0163$ \\
    & XDGMM   & $1.8593 \pm 0.0310$ & $0.4643 \pm 0.0114$ \\
    & Oracle  & \textit{0.8217 $\pm$ 0.0191} & \textit{0.0801 $\pm$ 0.0233}\\ \midrule
\multirow{5}{*}{Checkerboard} 
    & convMMD\_NF & $\mathbf{0.4383 \pm 0.0196}$ & $0.1785 \pm 0.0134$ \\
    & convMMD\_GMM & $0.4424 \pm 0.0230$ & $\mathbf{0.1197 \pm 0.0071}$ \\
    & NPEB    & $0.5096 \pm 0.0186$ & $0.2763 \pm 0.0131$ \\
    & XDGMM   & $0.5260 \pm 0.0194$ & $0.2053 \pm 0.0064$ \\
    & Oracle  & \textit{0.3929 $\pm$ 0.0147} & \textit{0.0633 $\pm$ 0.0039} \\ \bottomrule
\end{tabular}
}
\vspace{-4mm}

\end{wraptable}

Figure~\ref{fig:2d-exp} shows denoised posterior means from a representative run. NPEB places the prior on a discrete grid, producing clustered, speckled estimates that miss the continuity of the underlying manifolds. XDGMM utilizes a continuous Gaussian mixture, but its EM-based likelihood optimization leads to oversmoothing, failing to recover sharp boundaries (Moons) or hollow geometries (Circle). By evaluating our framework with both Normalizing Flow (\texttt{convMMD\_NF}) and GMM (\texttt{convMMD\_GMM}) sieves, we isolate the advantage of the convMMD objective. Notably, despite sharing the same parametric family as XDGMM, \texttt{convMMD\_GMM} gives better reconstructions, confirming that simulation-based convMMD framework recovers complex topology better than likelihood maximization in presence of outliers. This observation is consistent with theoretical guarantees regarding MMD's robustness to corrupted mass \citep{cherief2022finite, alquier2024universal}. The sieves also exhibit distinct inductive biases: the continuous mapping of NF excels on continuous topological manifolds (Moons, Circles), while GMM efficiently captures multimodal clustered grids (Checkerboard).  We also investigate the sensitivity of convMMD to the choice of kernel bandwidth, as the median heuristic introduces an implicit free parameter. Results across a nine-point logarithmic grid of scale factors confirm that the performance remains stable over a wide range; see Appendix \ref{sec:bw_sensivity} and Figure \ref{fig:bandwidth-sens-1} for the full analysis.

\begin{figure}[htbp]
    \centering
    \includegraphics[trim=0mm 8mm 0mm 0mm,clip,width=1\linewidth]{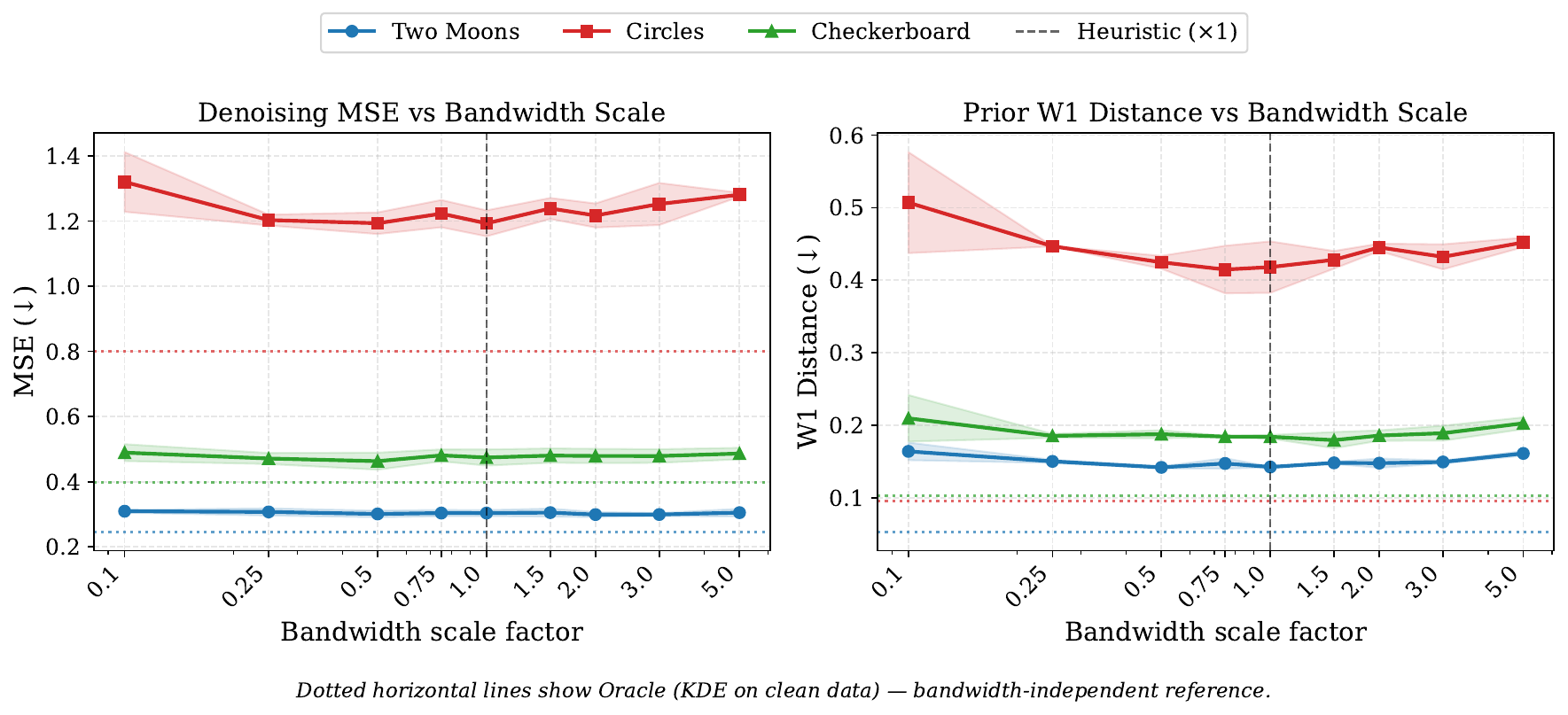}
    \caption{Dotted horizontal lines show Oracle (KDE on clean data) -- bandwidth independent reference. }
    \label{fig:bandwidth-sens-1}
\end{figure}

\paragraph{Noise Misspecification and Scalability across Dimensions.} We next evaluate the scalability and robustness to noise misspecification of convMMD for multivariate deconvolution. We use a ``Tangled Ribbon'' manifold: a low-intrinsic-dimensional topological structure embedded in observation spaces of increasing dimension, $D \in \{3,7,11,15\}$, and corrupted with heteroscedastic isotropic Gaussian noise, where per-sample noise scales $\sigma_i \sim \mathcal{U}(0.085, 0.127)$ with mean $\bar{\sigma} = 0.106$. We compare convMMD (with Normalizing Flows) against three multivariate deconvolution baselines: XDGMM \citep{bovy2011extreme}, NPEB \citep{soloff2025multivariate}, and deconv \citep{dockhorn2020density}. Full experimental details and network architectures are provided in Appendix~\ref{app:deconv_with_dim}.

We assess the robustness to noise misspecification by evaluating each method under three noise regimes: \emph{Well Specified}, where each method receives the true per-sample $\sigma_i$; \emph{Homoscedastic Misspecified}, where a single constant $\bar{\sigma}$ is provided; and \emph{Heteroscedastic Misspecified}, where per-sample estimates are corrupted by log-normal noise calibrated to yield 25\% expected relative error. We use the sliced 1-Wasserstein distance (SWD), computed over 1000 random one-dimensional projections, as our primary metric throughout, since exact Wasserstein distances scale poorly in high dimensions at rate $\mathcal{O}(n^{-1/D})$ \citep{fournier2015rate}. 

\begin{table}[htbp]
\centering
\caption{SWD (scaled by $\sqrt{D}$) between generated samples and the true manifold (see Figure~\ref{fig:deconv-3d-1}). Values are mean $\pm$ std.}
\label{tab:swd_results}
\resizebox{0.95\textwidth}{!}{
\begin{tabular}{ll cccc}
\toprule
\textbf{Noise Type} & \textbf{Method}
    & $D=3$ & $D=7$ & $D=11$ & $D=15$ \\
\midrule
\multirow{4}{*}{\makecell[l]{Well-specified\\(True $\sigma_i$)}}
  & convMMD
      & $0.0329 \pm 0.0049$ & $\mathbf{0.0339} \pm 0.0034$
      & $\mathbf{0.0342} \pm 0.0018$ & $\mathbf{0.0396} \pm 0.0051$ \\
  & XDGMM
      & $\mathbf{0.0290} \pm 0.0065$ & $0.0357 \pm 0.0051$
      & $0.0405 \pm 0.0025$ & $0.0488 \pm 0.0030$ \\
  & NPEB
      & $0.0351 \pm 0.0044$ & $0.0441 \pm 0.0033$
      & $0.0741 \pm 0.0052$ & $0.1161 \pm 0.0047$ \\
  & deconv
      & $0.0711 \pm 0.0058$ & $0.1006 \pm 0.0018$
      & $0.1317 \pm 0.0139$ & $0.2090 \pm 0.0346$ \\
\midrule
\multirow{4}{*}{\makecell[l]{Homoscedastic\\Misspecified}}
  & convMMD
      & $0.0334 \pm 0.0061$ & $\mathbf{0.0318} \pm 0.0017$
      & $\mathbf{0.0367} \pm 0.0036$ & $\mathbf{0.0363} \pm 0.0014$ \\
  & XDGMM
      & $\mathbf{0.0299} \pm 0.0047$ & $0.0376 \pm 0.0041$
      & $0.0422 \pm 0.0036$ & $0.0513 \pm 0.0029$ \\
  & NPEB
      & $0.0366 \pm 0.0038$ & $0.0453 \pm 0.0029$
      & $0.0707 \pm 0.0041$ & $0.1145 \pm 0.0055$ \\
  & deconv
      & $0.0770 \pm 0.0034$ & $0.1003 \pm 0.0025$
      & $0.1382 \pm 0.0201$ & $0.2105 \pm 0.0274$ \\
\midrule
\multirow{4}{*}{\makecell[l]{Heteroscedastic\\Misspecified}}
  & convMMD
      & $\mathbf{0.0318} \pm 0.0049$ & $\mathbf{0.0346} \pm 0.0036$
      & $\mathbf{0.0394} \pm 0.0032$ & $\mathbf{0.0469} \pm 0.0016$ \\
  & XDGMM
      & $0.0343 \pm 0.0037$ & $0.0457 \pm 0.0044$
      & $0.0544 \pm 0.0040$ & $0.0641 \pm 0.0047$ \\
  & NPEB
      & $0.0379 \pm 0.0036$ & $0.0516 \pm 0.0031$
      & $0.0767 \pm 0.0048$ & $0.1068 \pm 0.0032$ \\
  & deconv
      & $0.0807 \pm 0.0086$ & $0.1007 \pm 0.0032$
      & $0.1265 \pm 0.0038$ & $0.2192 \pm 0.0176$ \\
\bottomrule
\end{tabular}
}
\end{table}
Table~\ref{tab:swd_results} reports results across independent trials, and Figure~\ref{fig:deconv-3d-1} shows a representative example run for $D=3$.
In $D=3$ under oracle noise, both XDGMM and convMMD successfully recover the latent manifold, with XDGMM slightly more competitive at low dimensionality due to its inductive bias for multimodal densities. NPEB, as a discrete nonparametric estimator, also captures the manifold structure at $D=3$, but its performance degrades as dimension increases. As $D$ increases, convMMD outperforms all baselines by exploiting the expressivity of normalizing flows to recover the low-dimensional manifold structure and achieving significantly lower SWD at $D=11$ and $D=15$. deconv performs poorly across all settings, which we attribute to its sensitivity to sample size; in our experiments, it required substantially more than the $N=2000$ training points used here to produce well-calibrated deconvolved samples.

Under misspecification, convMMD remains the most robust method: its SWD increases only modestly from the oracle to the heteroscedastic misspecified regime across all dimensions. XDGMM shows greater sensitivity, particularly at higher dimensions; NPEB and deconv demonstrate poor performance. This robustness of convMMD to approximate noise specification is practically important, as exact per-sample noise may not be available in some real-world applications. 

\begin{figure}[htbp]
    \includegraphics[width=1\linewidth]{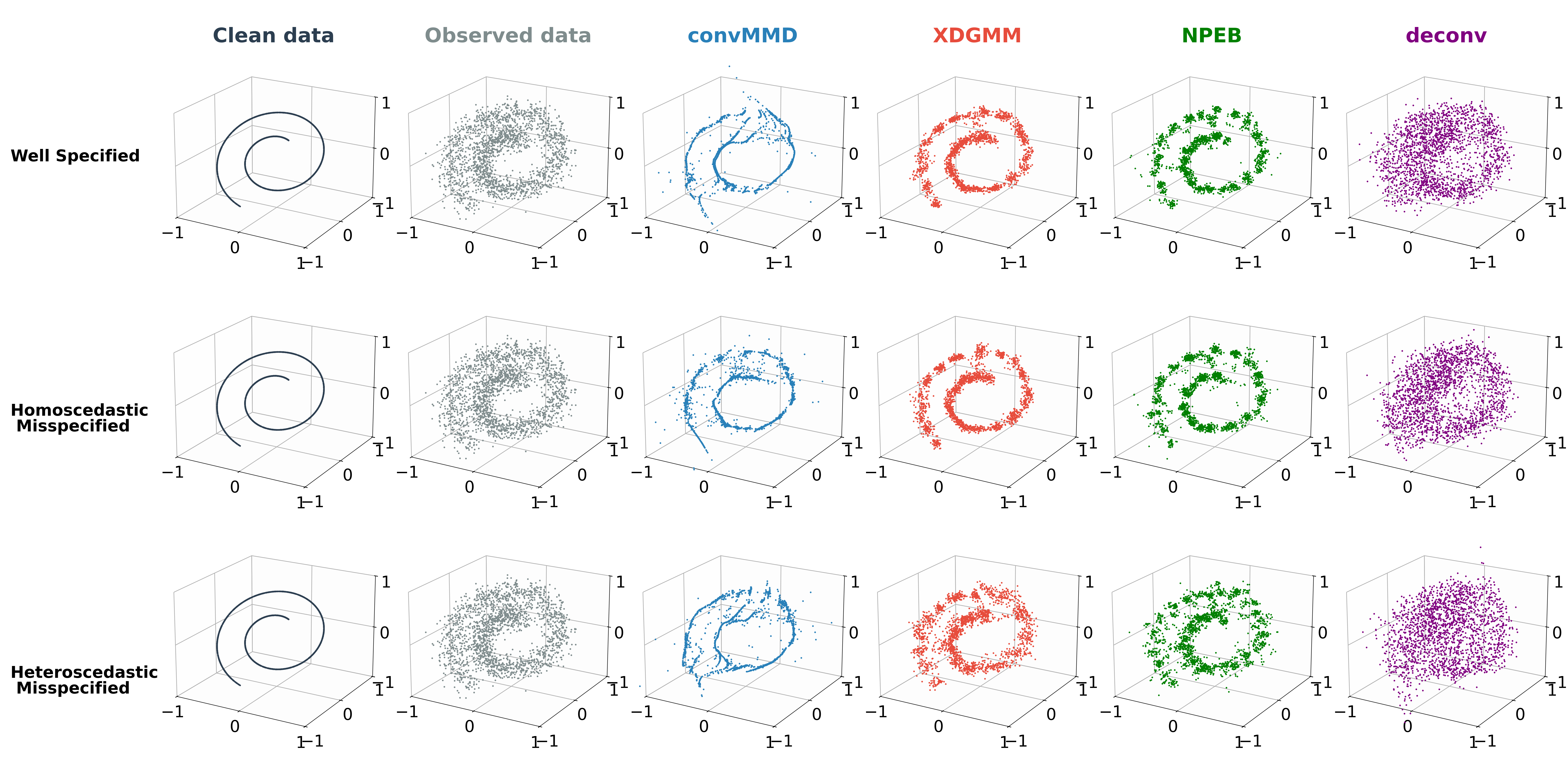}
    \caption{Qualitative Comparison of Estimated Prior for d = 3 across different methods (columns) and different noise models (rows). }
    \label{fig:deconv-3d-1}
\end{figure}

\paragraph{Performance in High-dimensional Denoising.}

In classical nonparametric statistics, density deconvolution is typically restricted to low-dimensional domains.  To illustrate scalability to higher-dimensional spaces ($D=784$), we apply our framework to MNIST image denoising. \begin{figure}[htbp]
    \centering
    \includegraphics[width=0.95\textwidth]{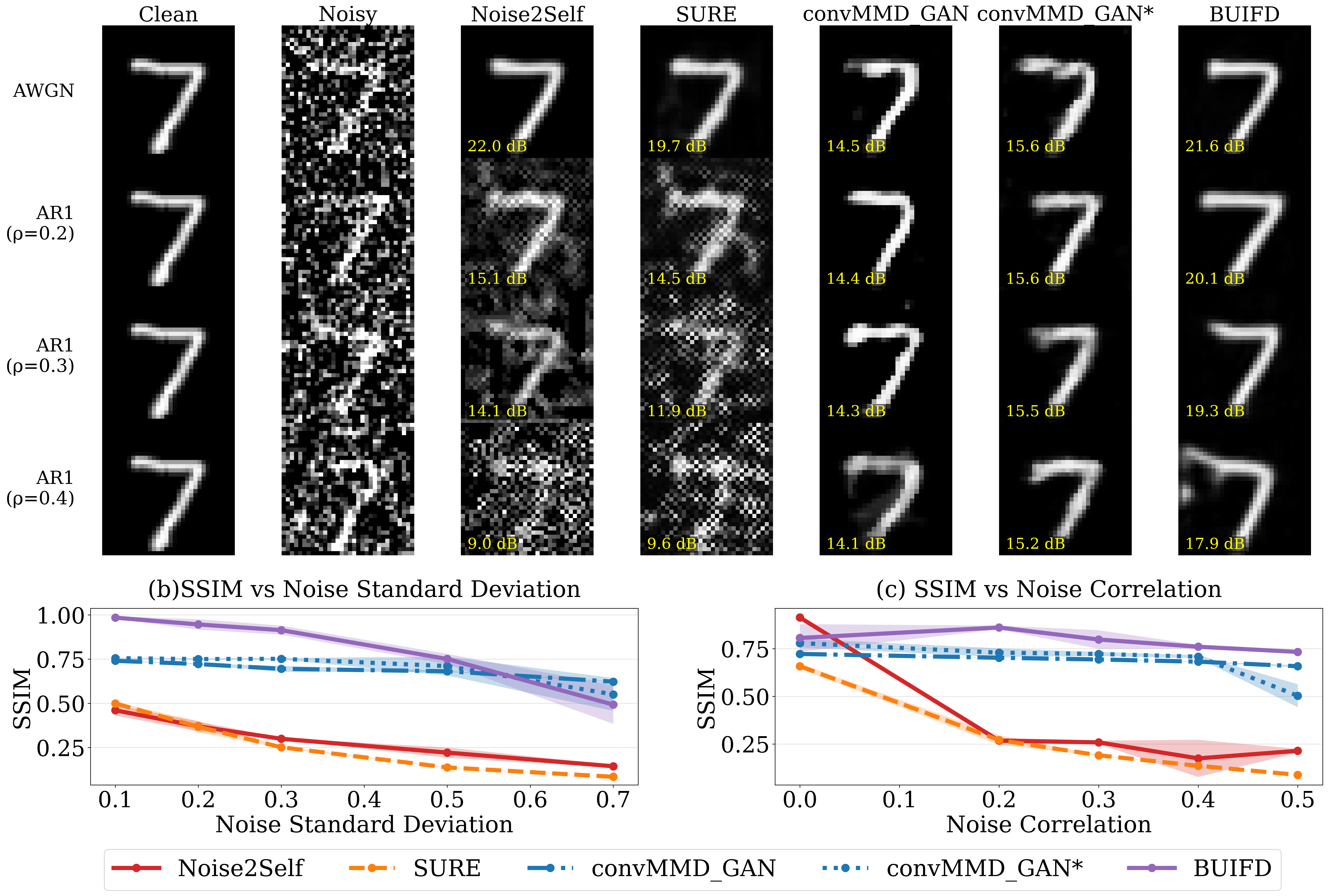}
    \caption{Image denoising performance.}
    \label{fig:mnist_comparison}
\end{figure} We consider unlabeled noisy MNIST images $\mathcal{D}_{\text{obs}}$, and use them to learn the prior $q_{\boldsymbol{\theta}}$ via convMMD and denoise via empirical Bayes. We parameterize $q_{\boldsymbol{\theta}}$ using two GAN variants: $\text{convMMD}\_\text{GAN}$ (GAN with a fixed Gaussian kernel) and $\text{convMMD}\_\text{GAN*}$ (GAN with an adversarially learned kernel \citep{li2017mmd}). We note that $\text{convMMD}\_\text{GAN*}$ may fall outside Assumptions \ref{assump:kernel} - \ref{assump:kernel-smooth}, however we include it as an empirical extension. Baselines include Noise2Self \citep{batson2019noise2self} (a self-supervised blind-spot network), SURE \citep{soltanayev2018training} (an unsupervised CNN assuming known Gaussian variance), and BUIFD \citep{el2020blind} (a supervised multi-branch model trained on paired data, serving as an empirical upper bound). We evaluate robustness under Additive White Gaussian Noise (AWGN), controlled by noise standard deviation $\sigma$, and spatially correlated AR(1) noise, controlled by correlation parameter $\rho \in [0, 1)$.
Figure \ref{fig:mnist_comparison} summarizes the results. The top panel illustrates denoising as $\rho$ increases (PSNR in yellow). Panels (b) and (c) report the Structural Similarity Index Measure (SSIM, \cite{wang2004image}) evaluated on the test set. We vary $\sigma$ at fixed $\rho=0.4$ (b), and vary the correlation $\rho$ at fixed $\sigma=0.5$ (c). Under pure AWGN ($\rho=0.0$), Noise2Self excels (22.0 dB) by exploiting pixel-wise noise independence. However, its performance collapses as spatial correlation increases (Fig. \ref{fig:mnist_comparison}a, c). In contrast, convMMD does not assume pixel independence; by matching the global convolved distribution, both variants remain remarkably stable (maintaining $\sim$15.2 dB at $\rho=0.4$, while Noise2Self drops to 9.0 dB). This robust stability persists across increasing noise intensities (Fig. \ref{fig:mnist_comparison}b), demonstrating that convMMD effectively scales classical deconvolution to complex, high-dimensional measurement error regimes. We note that although the high dimensionality of images may violate the global Sobolev smoothness assumption in Theorem \ref{th:l2rate}, our results demonstrate good performance. 
This experiment is intended as a stress test of the proposed deconvolution framework, rather than as a contribution to the image-denoising literature. Image denoising is a mature and highly specialized area, with methods that exploit image-specific inductive biases, carefully engineered architectures, perceptual losses, and noise models tailored to imaging applications \citep[e.g.,][]{buades2005review, dabov2007image, pu2016deep, ho2020denoising, luo2025taming}. By contrast, convMMD is not optimized for images and does not rely on the common assumptions used in this literature. Hence, comparisons with image-denoising methods should not be interpreted as a competition on an image-restoration benchmark. Rather, MNIST provides a transparent high-dimensional test bed in which the ambient dimension is substantially larger than in classical nonparametric deconvolution, and the noise can be made non-trivially correlated across features. The main lesson from this experiment is that distributional deconvolution via convMMD remains stable in a regime where methods relying on noise independence can degrade sharply. Thus, the purpose of the case study is to demonstrate scalability and robustness of the proposed measurement-error framework beyond the low-dimensional settings emphasized in the theory, while leaving image-specific architectural refinements and benchmark-level optimization to future work.

%% file: sections/conclusion.tex
\section{Conclusion} \label{sec:disc}
% \vspace{-3mm}
Our method, which combines flexibility and theory, contributes to the broader statistics and machine learning literature on deconvolution and denoising. Much of modern learning from corrupted data still relies either on heuristic denoisers with limited statistical interpretation or on classical estimators that do not align well with modern generative modeling practices. Our framework offers an alternative: train expressive latent models directly on corrupted observations with an objective that respects the known noise mechanism, supports scalable optimization, and admits both finite-sample and large-sample guarantees. In this sense, the paper advances deconvolution and denoising from a classical inverse problem framework toward a general paradigm for learning under structured corruption utilizing a likelihood-free, simulation-based framework.

\paragraph{Limitations.}  While the proposed framework offers a general approach to deconvolution and denoising, our current implementation involves several considerations. From a theoretical perspective, our convergence guarantees are established under specific regularity conditions (known noise distributions and Sobolev smoothness) and assume the attainment of global minima on a non-convex, Monte Carlo approximated objective. Extending these results to account for the dynamics of stochastic gradient methods and the sieve properties of modern neural architectures, like normalizing flows, while accommodating spatial and other dependence patterns in the data remains an area for further study.
In practice, the performance of convMMD is sensitive to the kernel and to the choice of the latent generative model and its architecture. Furthermore, the accuracy of the empirical Bayes denoising step is inherently linked to the quality of the learned latent density; consequently, model misspecification in high-dimensional or high-noise settings can impact posterior estimates. These factors highlight meaningful avenues for future research into data-driven kernel selection, choice of the latent model, and the robustness of the framework under noise misspecification.

%% file: sections/assumptions.tex
\section{Assumptions}
\label{app:assumptions}
For clarity, we collect here all assumptions used in the paper. Theorems in the main text cite the exact subset they require. 
\paragraph{Notation.} For two positive functions $f(\mathbf{t})$ and $g(\mathbf{t})$, we write $f(\mathbf{t}) \asymp g(\mathbf{t})$ to mean that there exist constants $0 < c \le C < \infty$ such that $c \, g(\mathbf{t}) \le f(\mathbf{t}) \le C \, g(\mathbf{t})$ for all $\mathbf{t}$ in the relevant domain. This two-sided bound is essential for the $L_2$ convergence analysis, as it allows us to both upper-bound the low-frequency error and lower-bound the spectral weight when inverting from MMD to $L_2$.
\subsection{Data-Generating Assumptions}

\begin{assumption}[Noise independence]
\label{assump:noise-indep}
The measurement error $\mathbf{U}_X$ is independent of the latent variable $\mathbf{X}$, i.e.,
\[
\mathbf{X} \perp\!\!\!\perp \mathbf{U}_X.
\]
\end{assumption}

\begin{assumption}[Known noise model]
\label{assump:known-noise}
The noise distribution $m$ is known. In heteroscedastic settings, we allow
\[
\mathbf{U}_{X,i} \sim r(\cdot \mid \phi_i),
\qquad
\phi_i \sim g(\cdot \mid \psi_k),
\]
where the conditional family $r$ and the parameter distribution $g$ are known, or the values $\phi_i$ are observed.
\end{assumption}

\begin{assumption}[Finite second moments of the noise]
\label{assump:noise-second-moment}
The noise variables have finite second moments. Specifically,
\[
\mathbb{E}[\|\mathbf{U}_X\|^2 \mid \phi] = \alpha(\phi) < \infty,
\]
and $\alpha(\phi)$ is integrable with respect to $g$.
\end{assumption}

\begin{assumption}[Convolution identifiability]
\label{assump:conv-invertibility}
The characteristic function of the noise,
\[
\phi_m(\mathbf{t})=\mathbb{E}[e^{i\mathbf{t}^\top \mathbf{U}_X}],
\]
is nonzero almost everywhere.
\end{assumption}

\begin{assumption}[Noise smoothness] \label{assump:noise-smooth}
  Let $\phi_m(\mathbf{t}) = \mathbb{E}[e^{i\mathbf{t}^\top \mathbf{U}_X}]$ be the characteristic function of the noise. We assume $\phi_m$ satisfies one of the following two-sided decay conditions for constants $\gamma > 0$:
  \begin{itemize}
      \item \textbf{(OS-N) Ordinary Smooth Noise}: $|\phi_m(\mathbf{t})| \asymp (1+\|\mathbf{t}\|^2)^{-\gamma/2}$.\\
      Examples: Laplace noise ($\gamma = 2$), Gamma distributions.
      \item \textbf{(SS-N) Supersmooth Noise}: $|\phi_m(\mathbf{t})| \asymp \exp(-c_m \|\mathbf{t}\|^\gamma)$ for some $c_m > 0$.\\
      Examples: Gaussian noise ($\gamma = 2$), Cauchy noise ($\gamma = 1$).
  \end{itemize}
  The two-sided bound ensures that the characteristic function decays at a known polynomial or exponential rate, neither faster nor slower.
\end{assumption}

\begin{assumption}[Latent Sobolev smoothness]
\label{assump:latent-sobolev}
The true latent density $p$ belongs to the Sobolev space $H^\beta(\mathbb{R}^d)$ for some $\beta > d/2$, that is,
\[
\int_{\mathbb{R}^d}
(1+\|\mathbf{t}\|^2)^\beta
|\phi_p(\mathbf{t})|^2\, d\mathbf{t}
<
\infty.
\]
\end{assumption}

\subsection{Model and Method Assumptions}

\begin{assumption}[Kernel bound]
\label{assump:kernel}
The kernel $k:\mathbb{R}^d \times \mathbb{R}^d \to \mathbb{R}$ is translation invariant, so that $k(\mathbf{x},\mathbf{y})=k(\mathbf{x}-\mathbf{y})$, and bounded:
\[
\sup_{\mathbf{x}} k(\mathbf{x},\mathbf{x}) \le K < \infty.
\]
Moreover, $k \in L^{1}(\mathbb{R}^d)$ and its Fourier transform $\phi_{k}(\mathbf{t})$ is strictly positive for all $\mathbf{t} \in \mathbb{R}^d$.
\end{assumption}

\begin{assumption}[Kernel smoothness] \label{assump:kernel-smooth}
  The Fourier transform of the kernel satisfies one of the following two-sided decay conditions for constants $\nu > 0$:
  \begin{itemize}
      \item \textbf{(OS-K) Ordinary Smooth Kernel}: $\phi_{k}(\mathbf{t}) \asymp (1+\|\mathbf{t}\|^2)^{-\nu/2}$.\\
      Examples: Laplace kernel ($\nu = d+1$), Mat\'ern kernels.
      \item \textbf{(SS-K) Supersmooth Kernel}: $\phi_{k}(\mathbf{t}) \asymp \exp(-c_k \|\mathbf{t}\|^\nu)$ for some $c_k > 0$.\\
      Examples: Gaussian kernel ($\nu = 2$).
  \end{itemize}
  The two-sided bound ensures that the kernel's spectral weight decays at a controlled rate, which is crucial for relating MMD convergence to $L_2$ convergence.
\end{assumption}

\begin{assumption}[Sieve approximation property]
\label{assump:sieve-approx}
Let $\mathcal{Q}_J=\{q_{\boldsymbol{\theta}}:\boldsymbol{\theta}\in\Theta_J\}$ be the sieve class. For every target density $p$ in the relevant smoothness class, there exists $q_J \in \mathcal{Q}_J$ such that
\[
\inf_{q \in \mathcal{Q}_J}
\|q-p\|_{L_2(\mathcal{X})}
\le
C_{\mathrm{approx}} J^{-\beta/d}
\]
for some constant $C_{\mathrm{approx}}>0$.
\end{assumption}

\begin{assumption}[Compact latent domain]
\label{assump:compact-domain}
The true latent distribution $p$ and the model distributions $q_{\boldsymbol{\theta}}$ for all $\boldsymbol{\theta} \in \Theta_J$ are supported on a compact domain $\mathcal{X} \subset \mathbb{R}^d$, such that $\|\mathbf{x}\| \le M$ for all $\mathbf{x} \in \mathcal{X}$ and some finite constant $M > 0$.
\end{assumption}
\begin{remark}
\label{rem:compact-support}
Assumption~\ref{assump:compact-domain} is a theoretical requirement for the Empirical Bayes denoising bounds, where the bound $M$ on the posterior mean appears explicitly. For Normalizing Flows and GMMs, which traditionally possess support on $\mathbb{R}^d$, the theoretical sieve class is assumed to be smoothly truncated to a sufficiently large compact domain $\mathcal{X}$ encompassing the data manifold. 
\end{remark}

\begin{assumption}[Parameter space compactness] \label{assump:compactness}
  For a fixed sieve size $J$, the parameters $\boldsymbol{\theta}$ are restricted to a compact subset $\Theta_J \subset \mathbb{R}^{D_J}$. Additionally, the optimization is
   regularized such that the estimator maintains a uniformly bounded Sobolev norm: $\sup_{\boldsymbol{\theta} \in \Theta_J} \|q_{\boldsymbol{\theta}}\|_{H^\beta} \leq C_B < \infty$.
\end{assumption}

\begin{assumption}[Bounded noise density]
\label{assump:bounded-noise}
The noise density satisfies $\|m\|_\infty < \infty$.
\end{assumption}

\begin{assumption}[Strictly positive marginal]
\label{assump:positive-marginal}
The true marginal density $\widetilde{p} = p * m$ is strictly positive on a compact observation domain $\mathcal{K} \subset \mathbb{R}^d$, i.e., $\inf_{\widetilde{\mathbf{x}} \in \mathcal{K}} \widetilde{p}(\widetilde{\mathbf{x}}) > 0$.
\end{assumption}

Assumptions~\ref{assump:noise-indep}--\ref{assump:positive-marginal} are organized in increasing strength. Assumptions~\ref{assump:noise-indep}, \ref{assump:known-noise}, \ref{assump:noise-second-moment}, \ref{assump:conv-invertibility}, \ref{assump:kernel}, and \ref{assump:compactness} are sufficient for the basic convMMD guarantees in the main text. Assumptions~\ref{assump:noise-smooth}, \ref{assump:latent-sobolev}, and \ref{assump:sieve-approx} are additionally needed for the strong nonparametric $L_2$ rates. For the empirical Bayes denoising bounds, Assumptions~\ref{assump:bounded-noise} and \ref{assump:positive-marginal} are required in addition to the $L_2$ rate assumptions.

\section{Optimization Algorithm}
\label{app:algorithm}

The empirical convMMD objective in Equation~\eqref{eq:objective} is generally non-convex, and the expectations appearing in it are typically unavailable in closed form. We therefore optimize it by Monte Carlo approximation and automatic differentiation. The essential requirement is that the model sampling mechanism and the noise simulation can be written as differentiable transformations of exogenous base randomness. This allows the use of pathwise gradients and standard stochastic optimization methods such as Adam.

\begin{algorithm}[ht!]
\caption{Sieve-convMMD density estimation via Monte Carlo and AutoDiff}
\label{alg:optimization}
\begin{algorithmic}[1]
\Statex \textbf{Input:} noisy data $\mathcal{D}_N=\{\widetilde{\mathbf{x}}_i\}_{i=1}^N$, sieve class $\mathcal{Q}_J=\{q_{\boldsymbol{\theta}}:\boldsymbol{\theta}\in\Theta_J\}$, initial parameter $\boldsymbol{\theta}^{(0)}$, kernel $k$, learning rate $\eta$, Monte Carlo batch size $M$, total iterations $N_{\mathrm{iter}}$
\Statex \textbf{Output:} estimator $\widehat{\boldsymbol{\theta}}_N$

\For{$t=1,\dots,N_{\mathrm{iter}}$}
    \State Sample a minibatch $\{\widetilde{\mathbf{x}}_b\}_{b=1}^M \subset \mathcal{D}_N$.
    \State Sample base latent randomness $\mathbf{Z}_1,\dots,\mathbf{Z}_M \sim p_z$.
    \State Sample measurement noise variables $\mathbf{U}_{Y,1},\dots,\mathbf{U}_{Y,M} \sim m$.
    \State Generate latent model samples $\mathbf{Y}_b(\boldsymbol{\theta}^{(t-1)}) \sim q_{\boldsymbol{\theta}^{(t-1)}}$ via reparameterization.
    \State Form noisy model samples $\widetilde{\mathbf{Y}}_b(\boldsymbol{\theta}^{(t-1)}) = \mathbf{Y}_b(\boldsymbol{\theta}^{(t-1)}) + \mathbf{U}_{Y,b}$.
    \State Compute a Monte Carlo estimate $\widehat{L}_M(\boldsymbol{\theta}^{(t-1)})$ of the empirical convMMD objective.
    \State Compute the pathwise gradient
    \[
        \widehat{\mathbf{g}}^{(t)}
        =
        \nabla_{\boldsymbol{\theta}} \widehat{L}_M(\boldsymbol{\theta})\big|_{\boldsymbol{\theta}=\boldsymbol{\theta}^{(t-1)}}.
    \]
    \State Update parameters
    \[
        \boldsymbol{\theta}^{(t)}
        =
        \boldsymbol{\theta}^{(t-1)} - \eta \widehat{\mathbf{g}}^{(t)}.
    \]
\EndFor

\State \textbf{return} $\widehat{\boldsymbol{\theta}}_N = \boldsymbol{\theta}^{(N_{\mathrm{iter}})}$
\end{algorithmic}
\end{algorithm}

\paragraph{Normalizing Flows:}
For models like Normalizing Flows, we sample base noise $\mathbf{Z} \sim p_z$ and measurement noise $\mathbf{U}_Y \sim m$, and generate the noisy observation as a differentiable computational graph:
$$ \widetilde{\mathbf{Y}}(\boldsymbol{\theta}) = G_{\boldsymbol{\theta}}(\mathbf{Z}) + \mathbf{U}_Y $$
We then substitute $\widetilde{\mathbf{Y}}(\boldsymbol{\theta})$ directly into the empirical MMD estimator $\widehat{L}_M(\boldsymbol{\theta})$ and backpropagate.

\paragraph{Gaussian Mixture Models:}
Applying the reparameterization trick to Gaussian Mixture Models is traditionally hindered by the discrete sampling of the mixture components $c \sim \text{Categorical}(\boldsymbol{\pi})$. However, because the MMD objective is an expectation, we can analytically marginalize the discrete choices, pulling the scalar weights $\pi_c$ outside the integral. Let $\mathbf{Y}_c(\boldsymbol{\theta}_c) = \boldsymbol{\mu}_c + \boldsymbol{\Sigma}_c^{1/2} \boldsymbol{\epsilon}$ where $\boldsymbol{\epsilon} \sim \mathcal{N}(\mathbf{0}, \mathbf{I}_d)$. The expected cross-term between the observed data and the model becomes:
$$ \mathbb{E}_{\widetilde{\mathbf{Y}} \sim q_{\boldsymbol{\theta}} * m} [k(\widetilde{\mathbf{x}}_i, \widetilde{\mathbf{Y}})] = \sum_{c=1}^J \pi_c \, \mathbb{E}_{\boldsymbol{\epsilon}, \mathbf{U}_Y}\left[k\left(\widetilde{\mathbf{x}}_i, \boldsymbol{\mu}_c + \boldsymbol{\Sigma}_c^{1/2}\boldsymbol{\epsilon} + \mathbf{U}_Y\right)\right] $$
Similarly, the $U$-statistic for the model-to-model term becomes a double summation over the components $c$ and $d$, weighted by $\pi_c \pi_d$. By simulating batches of $\boldsymbol{\epsilon}$ and $\mathbf{U}_Y$ for each component, we construct a fully differentiable empirical loss where gradients flow deterministically into $\boldsymbol{\mu}_c$ and $\boldsymbol{\Sigma}_c$ via reparameterization, and analytically into the mixture weights $\pi_c$.

%% file: sections/justifications.tex
\section{Verification of Model Classes as Sieve Classes}
\label{app:model_classes}

In this section, we formally verify how our two representative model families—Gaussian Mixture Models and Neural Generative Models —satisfy the structural and theoretical assumptions of the convMMD framework (Assumptions~\ref{assump:compactness} and \ref{assump:sieve-approx}), and the self-regularization condition required for empirical Bayes denoising.

\subsection{Gaussian Mixture Models (GMMs)}
\label{app:gmm-justification}

Gaussian mixtures serve as our mathematically explicit baseline for which the entire chain of assumptions, from basic compactness to quantitative $L_2$ approximation, can be rigorously verified.

\paragraph{Parameter Space and Compactness.}
For a fixed complexity $J \ge 1$, we define the sieve class $\mathcal{Q}_J = \{q_{\boldsymbol{\theta}}(\mathbf{x}) = \sum_{c=1}^J \pi_c \, \mathcal{N}(\mathbf{x} \mid \boldsymbol{\mu}_c, \boldsymbol{\Sigma}_c) : \boldsymbol{\theta} \in \Theta_J\}$. To satisfy the compactness and regularization requirements of Assumption~\ref{assump:compactness}, we consider the parameter space:
\begin{equation}
\Theta_J = \Delta_J^\epsilon \times \mathcal{M}^J \times \mathcal{S}(\underline{\lambda}, \overline{\lambda})^J,
\label{eq:gmm-compact-space}
\end{equation}
where $\Delta_J^\epsilon = \{\pi \in [0,1]^J : \sum \pi_c = 1, \pi_c \ge \epsilon/J\}$ is the bounded-away-from-zero simplex, $\mathcal{M} = [-M_\mu, M_\mu]^d$ restricts the means, and $\mathcal{S}(\underline{\lambda}, \overline{\lambda}) = \{\boldsymbol{\Sigma} \in \mathbb{S}_{++}^d : \underline{\lambda} \mathbf{I}_d \preceq \boldsymbol{\Sigma} \preceq \overline{\lambda} \mathbf{I}_d\}$ bounds the covariance eigenvalues. 

% Because $\Theta_J$ is a closed and bounded subset of a finite-dimensional Euclidean space, it is strictly compact. The map $\boldsymbol{\theta} \mapsto q_{\boldsymbol{\theta}}$ is continuous in both $L_1$ and arbitrary Sobolev norms $H^\beta$. Furthermore, the condition $\boldsymbol{\Sigma}_c \succeq \underline{\lambda} \mathbf{I}_d$ prevents the mixture components from collapsing into Dirac deltas, ensuring a uniform Sobolev bound $\sup_{\boldsymbol{\theta} \in \Theta_J} \|q_{\boldsymbol{\theta}}\|_{H^\beta} < \infty$, fully satisfying Assumption~\ref{assump:compactness}. We satisfy this requirement by adding an inverse-variance barrier penalty to the empirical objective. Specifically, we regularize the diagonal elements of the component covariances using a penalty proportional to $\sum_{c,k} 1/(\sigma_{c, k}^2 + \epsilon)$. Inspired by the trace penalty of an Inverse-Wishart prior, this soft constraint strictly prevents any component's variance from collapsing to zero. This ensures that the components do not degenerate into Dirac deltas, directly satisfying the uniform Sobolev regularity of Assumption~\ref{assump:compactness} and maintaining the self-regularization condition required for empirical Bayes denoising.

Because $\Theta_J$ is a closed and bounded subset of a finite-dimensional Euclidean space, it is strictly compact. The map $\boldsymbol{\theta} \mapsto q_{\boldsymbol{\theta}}$ is continuous in both $L_1$ and arbitrary Sobolev norms $H^\beta$. Furthermore, the condition $\boldsymbol{\Sigma}_c \succeq \underline{\lambda} \mathbf{I}_d$ prevents the mixture components from collapsing into Dirac deltas, ensuring a uniform Sobolev bound $\sup_{\boldsymbol{\theta} \in \Theta_J} \|q_{\boldsymbol{\theta}}\|_{H^\beta} < \infty$, fully satisfying Assumption~\ref{assump:compactness}.

In practice, we enforce the constraint $\boldsymbol{\Sigma}_c \succeq \underline{\lambda} \mathbf{I}_d$ by parameterizing each covariance as $\boldsymbol{\Sigma}_c = \mathbf{L}_c \mathbf{L}_c^\top + \underline{\lambda} \mathbf{I}_d$, where $\mathbf{L}_c$ is an unconstrained lower-triangular matrix. This Cholesky-based parameterization guarantees positive definiteness and ensures that all eigenvalues are at least $\underline{\lambda}$, preventing singularity in all dimensions $d \ge 1$. For computational efficiency, when restricting to diagonal covariances, we add a log-barrier penalty $-\alpha \sum_{c,k} \log(\sigma_{c,k}^2)$ to the objective, which implicitly enforces $\sigma_{c,k}^2 > 0$ during optimization.

\paragraph{Quantitative Sieve Approximation.}
The strong deconvolution rates in Theorem~\ref{th:l2rate} rely on the polynomial approximation capacity of the sieve (Assumption~\ref{assump:sieve-approx}). For GMMs, this property is well-established. Following the classical approximation theory of \citep{genovese2000rates} and \citep{ghosal2001entropies}, Gaussian mixtures with regularized parameters achieve polynomial approximation rates for sufficiently smooth target densities on compact domains. Specifically, for $p \in H^\beta(\mathcal{X})$, there exists a constant $C$ such that:
\begin{equation}
\inf_{\boldsymbol{\theta} \in \Theta_J} \|q_{\boldsymbol{\theta}} - p\|_{L_2(\mathcal{X})} \le C J^{-\beta/d}.
\end{equation}

\begin{remark} Achieving fine approximation as $J \to \infty$ requires concentrated components ($\underline{\lambda}_J \to 0$), yet individual sharp Gaussians have unbounded Sobolev norms, which violates our assumptions.

We resolve this by constraining the \emph{global mixture's} Sobolev norm rather than individual component eigenvalues:
\begin{equation}
\Theta_J = \left\{ \boldsymbol{\theta} : \boldsymbol{\Sigma}_c \succeq \underline{\lambda}_J \mathbf{I}_d \quad \text{and} \quad \|q_{\boldsymbol{\theta}}\|_{H^\beta} \le M \right\},
\end{equation}
where $\underline{\lambda}_J \asymp J^{-2/d}$ shrinks with $J$, and $M$ is a fixed constant independent of $J$.

The constraint $\|q_{\boldsymbol{\theta}}\|_{H^\beta} \le M$ prevents pathological configurations (e.g., isolated spikes) while permitting the fine-scale resolution needed for good approximation. 
\end{remark}

\subsection{Neural Generative Models}
\label{app:neural-gen-details}

We analyze normalizing flows, GANs, and diffusion models uniformly as pushforward distributions $q_{\boldsymbol{\theta}} = (G_{\boldsymbol{\theta}})_\# p_z$, where $\mathbf{Z} \sim p_z$ is a fixed base distribution and $G_{\boldsymbol{\theta}} : \mathbb{R}^k \to \mathbb{R}^d$ is a parameterized neural network. 

\paragraph{Handling Degenerate Densities (GANs).}
A theoretical concern with neural generative models (e.g., GANs with $k < d$) is that the latent pushforward $(G_{\boldsymbol{\theta}})_\# p_z$ may be supported on a lower-dimensional manifold, failing to admit a Lebesgue density. The convMMD framework cleanly circumvents this. 

\begin{proposition}[Existence of the convolved density]
\label{prop:convolved-density}
Let $m$ be a noise distribution admitting a bounded density $m \in L^\infty(\mathbb{R}^d)$. 
Then, for any pushforward measure $q_{\boldsymbol{\theta}} = (G_{\boldsymbol{\theta}})_\# p_z$ (regardless of whether $q_{\boldsymbol{\theta}}$ has a density), the convolved distribution $q_{\boldsymbol{\theta}} * m$ admits a Lebesgue density on $\mathbb{R}^d$ given by:
\begin{equation}
    (q_{\boldsymbol{\theta}} * m)(\widetilde{\mathbf{x}}) = \mathbb{E}_{\mathbf{Z} \sim p_z}[m(\widetilde{\mathbf{x}} - G_{\boldsymbol{\theta}}(\mathbf{Z}))].
\end{equation}
\end{proposition}
\begin{proof}
By Fubini's theorem, for any measurable set $A \subset \mathbb{R}^d$, the probability measure of the convolution is $(q_{\boldsymbol{\theta}} * m)(A) = \int \mathbf{1}_A(\widetilde{\mathbf{x}}) \left( \int m(\widetilde{\mathbf{x}} - \mathbf{y}) dq_{\boldsymbol{\theta}}(\mathbf{y}) \right) d\widetilde{\mathbf{x}}$. Since $m \in L^\infty$, the inner integral is a bounded measurable function defining the Radon--Nikodym derivative. 
\end{proof}
Consequently, the MMD objective in Theorem~\ref{th:oracle} is always well-defined, and optimization via reparameterization gradients $\widetilde{\mathbf{Y}} = G_{\boldsymbol{\theta}}(\mathbf{Z}) + \mathbf{U}_Y$ is unconditionally valid. For pure deconvolution tasks, evaluating the learned density is not possible if the output is degenerate. However, for empirical Bayes denoising, this is inconsequential: the posterior mean can still be accurately estimated using importance sampling drawn from the base measure $p_z$, without ever evaluating the latent density.

\paragraph{Compactness and Regularity.}
For a network $G_{\boldsymbol{\theta}}$ with $D_J$ parameters, compactness (Assumption~\ref{assump:compactness}) is enforced using standard regularizations: explicit weight clipping ($\|\mathbf{W}\|_\infty \le B$), or spectral normalization ($\|\mathbf{W}\|_{\mathrm{op}} \le s$). Under these constraints, $\Theta_J$ is a compact subset of $\mathbb{R}^{D_J}$. If spectral normalization is used, $G_{\boldsymbol{\theta}}$ is uniformly Lipschitz \citep{miyato2018spectral}. For diffeomorphic mappings like normalizing flows, bounded network weights and invertible architectures (e.g., affine coupling layers) ensure uniform bounds on the Jacobian, guaranteeing $\sup_{\boldsymbol{\theta} \in \Theta_J} \|q_{\boldsymbol{\theta}}\|_{H^\beta} < \infty$.

\paragraph{The Sieve Approximation Requirement.}
While these models satisfy the basic estimation guarantees (Theorem~\ref{th:oracle}), their applicability to the explicit $L_2$ deconvolution rates (Theorem~\ref{th:l2rate}) depends on Assumption~\ref{assump:sieve-approx}. 
Currently, strong universality results exist for these architectures: normalizing flows are universal diffeomorphism approximators \citep{teshima2020coupling}, and GANs can approximate any continuous distribution in Wasserstein distance \citep{liang2021well}. Furthermore, unconstrained ReLU networks achieve explicit polynomial approximation rates for Sobolev functions \citep{yarotsky2017error}. 

However, mapping the $L_\infty$ or $L_p$ function-approximation rates of a neural network directly to the $L_2$ approximation rate of its \emph{induced pushforward density} remains an active area of approximation theory. Thus, while implicit models are fully consistent under our convMMD framework, the exact polynomial deconvolution rates of Theorem~\ref{th:l2rate} apply to them \emph{conditionally}, assuming the underlying architecture is sized appropriately to satisfy the quantitative bound in Assumption~\ref{assump:sieve-approx}.

\paragraph{Empirical Bayes Implementation.}
For denoising, the empirical Bayes estimator requires the computation of the posterior mean. 
For normalizing flows, the exact density evaluation $q_{\boldsymbol{\theta}}(\mathbf{x}) = p_z(G_{\boldsymbol{\theta}}^{-1}(\mathbf{x})) |\det \nabla G_{\boldsymbol{\theta}}^{-1}(\mathbf{x})|$ allows direct numerical integration or MCMC sampling. 
For GANs and diffusion models lacking tractable densities, we bypass exact evaluation by using Self-Normalized Importance Sampling (SNIS) in the latent space:
\begin{equation}
    \widehat{\mathbf{x}}^{\mathrm{EB}}(\widetilde{\mathbf{x}}_i) \approx \frac{\sum_{j=1}^S G_{\boldsymbol{\theta}}(\mathbf{z}_j) \, m(\widetilde{\mathbf{x}}_i - G_{\boldsymbol{\theta}}(\mathbf{z}_j))}{\sum_{j=1}^S m(\widetilde{\mathbf{x}}_i - G_{\boldsymbol{\theta}}(\mathbf{z}_j))}, \qquad \mathbf{z}_j \stackrel{\mathrm{iid}}{\sim} p_z.
\end{equation}
This estimator converges to the true posterior mean $\mathbb{E}_{q_{\boldsymbol{\theta}}}[\mathbf{X} \mid \widetilde{\mathbf{X}} = \widetilde{\mathbf{x}}_i]$ as $S \to \infty$ and leverages only the forward pass of the generator and the known noise density $m$, fitting seamlessly into the convMMD pipeline.